\documentclass[sigconf]{acmart}

\usepackage{booktabs} 

\copyrightyear{2023}
\acmYear{2023}
\setcopyright{acmlicensed}
\acmConference[CIKM '23] {Proceedings of the 32nd ACM International Conference on Information and Knowledge Management}{October 21--25, 2023}{Birmingham, United Kingdom.}
\acmBooktitle{Proceedings of the 32nd ACM International Conference on Information and Knowledge Management (CIKM '23), October 21--25, 2023, Birmingham, United Kingdom}
\acmPrice{15.00}
\acmISBN{979-8-4007-0124-5/23/10}
\acmDOI{10.1145/3583780.3615031}

\usepackage{amsthm}
\usepackage{amscd,amsfonts,amsbsy,rotating}
\usepackage{pifont}
\usepackage{balance}
\usepackage{graphicx}
\usepackage{epsfig,epstopdf}
\usepackage{subfigure}
\usepackage{multirow}
\usepackage{booktabs}
\usepackage{color,xcolor}
\usepackage{url}
\usepackage{latexsym,bm}
\usepackage{enumitem,balance,mathtools}
\usepackage{wrapfig}
\usepackage{euscript}
\usepackage{ifpdf}
\usepackage{diagbox}
\usepackage{caption}
\usepackage{makecell}
\usepackage{subfigure}
\usepackage[leftcaption]{sidecap}
\usepackage{textcomp}
\usepackage[normalem]{ulem}
\usepackage{breakurl}
\usepackage{diagbox}
\usepackage{array}
\newcolumntype{N}{@{}m{0pt}@{}}

\usepackage{lipsum}
\PassOptionsToPackage{hyphens}{url}
\usepackage{hyperref}

\usepackage[ruled,linesnumbered]{algorithm2e}

\usepackage{algorithmic}

\newcommand{\minisection}[1]{\vspace{5pt}\noindent\textbf{#1}}

\begin{document}
\title{Replace Scoring with Arrangement: A Contextual Set-to-Arrangement Framework for Learning-to-Rank}
\author{Jiarui Jin$^{*}$}
\affiliation{
\institution{Shanghai Jiao Tong University}
\country{Shanghai, China}
}
\email{jinjiarui97@sjtu.edu.cn}

\author{Xianyu Chen$^{*}$}
\affiliation{
\institution{Shanghai Jiao Tong University}
\country{Shanghai, China}
}
\email{xianyujun@sjtu.edu.cn}

\author{Weinan Zhang}
\affiliation{
\institution{Shanghai Jiao Tong University}
\country{Shanghai, China}
}
\email{wnzhang@sjtu.edu.cn}

\author{Mengyue Yang}
\affiliation{
\institution{University College London}
\country{London, United Kingdom}
}
\email{mengyue.yang.20@ucl.ac.uk}

\author{Yang Wang$^{\dag}$}
\affiliation{
\institution{East China Normal University}
\country{Shanghai, China}
}
\email{ywang@sei.ecnu.edu.cn}

\author{Yali Du}
\affiliation{
\institution{King's College London}
\country{London, United Kingdom}
}
\email{yali.du@kcl.ac.uk}

\author{Yong Yu}
\affiliation{
\institution{Shanghai Jiao Tong University}
\country{Shanghai, China}
}
\email{yyu@sjtu.edu.cn}

\author{Jun Wang}
\affiliation{
\institution{University College London}
\country{London, United Kingdom}
}
\email{jun.wang@ucl.ac.uk}

\renewcommand{\shortauthors}{Jiarui Jin et al.}
\renewcommand{\shorttitle}{STARank}

\ccsdesc[500]{Information systems~Learning to rank}
\ccsdesc[500]{Information systems~Personalization}

\settopmatter{printacmref=false}

\begin{abstract}
Learning-to-rank is a core technique in the top-N recommendation task, where an ideal ranker would be a mapping from an item set to an arrangement (a.k.a. permutation).
Most existing solutions fall in the paradigm of probabilistic ranking principle (PRP), i.e., first score each item in the candidate set and then perform a sort operation to generate the top ranking list.
However, these approaches neglect the contextual dependence among candidate items during individual scoring, and the sort operation is non-differentiable.
To bypass the above issues, we propose \underline{\textsf{S}}et-\underline{\textsf{T}}o-\underline{\textsf{A}}rrangement \underline{\textsf{Rank}}ing (\textsf{STARank}), a new framework directly generates the permutations of the candidate items without the need for individually scoring and sort operations; and is end-to-end differentiable. 
As a result, \textsf{STARank} can operate when only the ground-truth permutations are accessible without requiring access to the ground-truth relevance scores for items.
For this purpose, \textsf{STARank} first reads the candidate items in the context of the user browsing history, whose representations are fed into a Plackett-Luce module to arrange the given items into a list.
To effectively utilize the given ground-truth permutations for supervising \textsf{STARank}, we 
leverage the internal consistency property of Plackett-Luce models to derive a computationally efficient list-wise loss.
Experimental comparisons against 9 the state-of-the-art methods on 2 learning-to-rank benchmark datasets and 3 top-N real-world recommendation datasets demonstrate the superiority of \textsf{STARank} in terms of conventional ranking metrics.
Notice that these ranking metrics do not consider the effects of the contextual dependence among the items in the list, we design a new family of simulation-based ranking metrics, where existing metrics can be regarded as special cases. 
\textsf{STARank} can consistently achieve better performance in terms of PBM and UBM simulation-based metrics.
\end{abstract}

\keywords{Learning-to-Rank, Contextual Set-to-Arrangement, Efficient Supervision Generation}
\settopmatter{printacmref=false} 
\maketitle

{
\renewcommand{\thefootnote}{\fnsymbol{footnote}}
\footnotetext[1]{Equal contributions: Jiarui Jin and Xianyu Chen.}
\footnotetext[2]{Correspondence: Yang Wang.}
}

{\fontsize{8pt}{8pt} \selectfont
	\textbf{ACM Reference Format:}\\
	Jiarui Jin, Xianyu Chen, Weinan Zhang, Mengyue Yang, Yang Wang, Yali Du, Yong Yu, Jun Wang. 2023. Replace Scoring with Arrangement: A Contextual Set-to-Arrangement Framework for Learning-to-Rank. In \textit{Proceedings of the 32nd ACM International Conference on Information and Knowledge Management, October 21--25, 2023, Birmingham, United Kingdom}. ACM, New York, NY, USA, 10 pages.
	\url{https://doi.org/10.1145/3583780.3615031}}

\section{Introduction}
Learning-to-rank (LTR) covers a branch of machine learning methods for optimizing ranking performance and is a core technique for search engine and top-N recommendation applications 
\citep{liu2011learning}.
Conventional LTR algorithms are usually designed on the basis of the probability ranking principle (PRP) \citep{robertson1977probability}, which consists of a scoring function that assigns an individual score to each item, and subsequently produces the ranking list by sorting items according to their assigned scores.
Note that the crucial difference between LTR and classical machine learning tasks (e.g., regression, classification) is that in LTR, the exact score of each item is not important, but the relative orders of the items matter.
It is supported by the behavior analysis results on search engines \citep{joachims2017accurately,yilmaz2014relevance}, which manifests that user interactions show strong comparison patterns, namely users usually compare multiple items before generating a click action.

Being aware of PRP's limitation of independently assigning the ranking score for each item and using non-differentiable sort operations \citep{xia2008listwise}, there are roughly two lines of previous researches. 
One direction \citep{ai2018learning,pang2020setrank,bello2018seq2slate} investigates incorporating context into the ranking process and assigning scores to candidate items while taking into account their dependencies within the context.
The other direction \citep{cao2007learning,guiver2009bayesian,taylor2008softrank,xia2008listwise} designs a surrogate differentiable list-wise loss based on permutation probabilities proposed in \citep{plackett1975analysis} to avoid the use of the sort operation.
However, all these methods still fall into a paradigm of first individually scoring items and then putting them into correct positions in descending order of their scores.

In this paper, we propose a new \underline{\textsf{S}}et-\underline{\textsf{T}}o-\underline{\textsf{A}}rrangement \underline{\textsf{Rank}}ing (\textsf{STARank}) framework, which can directly generate the permutations of candidate items without the need for individually scoring and sort operations. 
Hence, our proposed framework is end-to-end differentiable, and can be directly supervised with ground-truth permutations instead of ground-truth scores for all candidate items. 
Concretely, as depicted in Figure~\ref{fig:overview}(a), the ultimate goal is to arrange the candidate items into a permutation close to the ground-truth permeation based on the historical records from users.
For this purpose, we develop a novel read-arrange-supervise paradigm.

Our reader component, as depicted in Figure~\ref{fig:overview}(b), is required to encode two types of input items, namely, the user browsed items and the candidate items.
Their major difference lies in that the user browsed items are characterized by their permutation orders, while the candidate items do not have an inherent permutation order.
Therefore, we first build a permutation-sensitive module to encode the browsed items, whose outputs are further fed into a permutation-invariant module to produce the representation of each candidate item.
Then, our arranger, as shown in Figure~\ref{fig:overview}(c), applies a Plackett-Luce (PL) module to match the candidate items to the remaining positions.

Consider that we only have access to the ground-truth permutations.
We leverage the internal consistency property of the arranger (i.e., PL models) to derive a list-wise loss function, as shown in Figure~\ref{fig:overview}(d).
Compared to existing loss functions based on the ground-truth relevance scores, our loss focuses on identifying and correcting the relative permutations of neighboring pairs that have errors during arranging.
In other words, this list-wise loss is typically designed for our recursive permutation generation. 
Our theoretical analysis also reveals that the proposed framework is permutation-invariant regarding the candidate items (namely any permutation of the input items would not change the output ranking results).

In addition, we notice that 
conventional ranking metrics overlook the effects of contextual dependence among the items (i.e., the probability of a user favoring an item is affected by other items placed in the same list).
Therefore, we develop a new series of simulation (i.e., click models) based metrics, where NDCG can be regarded as special cases of using PBM simulations \citep{richardson2007predicting}. 

We conduct experiments on 2 learning-to-rank benchmark datasets and 3 real-world recommendation datasets.
Empirical results demonstrate that \textsf{STARank} can consistently outperform 9 the state-of-art baseline methods in terms of both conventional ranking metrics and new proposed simulation-based metrics.

\section{Preliminaries }
\subsection{Problem Formulation}
\label{sec:ranker}
LTR algorithms usually assume that each item $d$ has its utility in the context of a user (or a query) $q$\footnote{In the context of LTR \citep{liu2011learning}, the terms ``query'' and ``user'' normally appear in the search engines and the top-N recommendation scenarios respectively. For simplicity, we mainly use ``user'' in this paper.}, 
which is often modelled as the relevance score between $q$ and $d$.
Let $\mathcal{D}$ denote the set of candidate items and $\mathcal{D}_q$ denote the set of candidate items associated with $q$.
Besides these candidate items to be arranged (a.k.a., permuted), we would also have access to user $q$'s historical records, a sequence of the browsed item list, denoted as $\mathcal{H}_q$.
Notably, $\mathcal{D}_q$ and $\mathcal{H}_q$ are subsets of $\mathcal{D}$ (i.e., $\mathcal{D}_q\subseteq \mathcal{D}$, $\mathcal{H}_q \subseteq \mathcal{D}$), and there exists a permutation order in each $\mathcal{H}_q$ but no permutation order in each $\mathcal{D}_q$.
We further use $r_d$ and $r^*_d$ to denote the predicted and ground-truth relevance of $d$ respectively.

Prevailing ranking methods formulate the ranking problem as to predict $r$ for each item, whose risk function is defined as
\begin{equation}
\label{eqn:riskpoint}
\mathcal{F}_\mathtt{point} = \sum_{q} \sum_{d\in \mathcal{D}_q} \ell_\mathtt{point}\Big(\mathtt{f}_\mathtt{point}(d|\mathcal{H}_q), r_d^*\Big) 
\end{equation}
where $\ell_\mathtt{point}(\cdot,\cdot)$ denotes a point-wise loss function, $\mathtt{f}_\mathtt{point}(d|\mathcal{H}_q)$ is a ranker that makes prediction over one item $d$ in the context of $\mathcal{H}_q$.
In practice, following PRP, the trained ranker is used to assign an individual score for each item independently and sort the items according to their scores.  
However, this pipeline internally suffers from the following limitations:
(i) it lacks the consideration over the contextual dependence among items, as the score of each item is solely determined by itself;
(ii) it requires the sort operation, which brings the non-differentiable issue.

\begin{table}[t]
    \caption{A summary of notations associated with $\pi$. 
    }  
    \label{tab:notation}
    \vspace{-4mm}
    \begin{center}  
    \begin{tabular}{p{2cm}<{\centering}p{6cm}<{\centering}}
    \toprule
    \textbf{Notations} & \textbf{Explanations} \\
    \midrule
    $\pi_i \in\mathcal{D}_q$ & The item at the $i$-th position\\
    \midrule
    $\pi_{<i}\subseteq \mathcal{D}_q$ & A set of the items at the higher positions than $i$ \\
    \bottomrule
    \end{tabular}  
    \end{center} 
    \vspace{-4mm}
\end{table}

Instead, we formulate the ranking problem as to seek for the ground-truth permutation for the given set of candidate items.
Formally, let $\pi(q)$ and $\pi^{*}(q)$ denote the predicted permutation and ground-truth permutation of the given item set $\mathcal{D}_q$ respectively.
In other words, $\pi(q)$ is the sequence of the arranged items in $\mathcal{D}_q$. 
In the followings, we use $\pi$ and $\pi^{*}$ instead of $\pi(q)$ and $\pi^{*}(q)$ for convenience.
We list the key notations associated with $\pi$ in Table~\ref{tab:notation}.

Our set-to-arrangement ranking problem can be described as
\begin{definition}
[\textbf{\emph{Set-to-Arrangement Ranking Problem}}]
\label{def:set}
For each user/query $q$, given a tuple $(\mathcal{D}_q, \mathcal{H}_q,\pi^*)$ where $\mathcal{D}_q$ is the set of candidate items and $\pi^*$ is the ground-truth permutation of $q$, and $\mathcal{H}_q$ is the set of browsed items; our goal is to
learn a ranker $\mathtt{f}_\mathtt{set}(\mathcal{D}_q|\mathcal{H}_q)$ that generates the permutation $\pi$ for the given set of items $\mathcal{D}_q$ in the context of $\mathcal{H}_q$.
The risk function can be written as  
\begin{equation}
\label{eqn:riskset}
\mathcal{F}_\mathtt{set} = \sum_{q} \ell_\mathtt{list}\Big(\mathtt{f}_\mathtt{set}(\mathcal{D}_q|\mathcal{H}_q), \pi^*\Big),
\end{equation}
where $\ell_\mathtt{list}(\cdot,\cdot)$ is a list-wise loss function.
\end{definition}
We illustrate our task in Figure~\ref{fig:overview}(a).
We note that in our task, only the ground-truth permutation $\pi^*$ is accessible, while the ground-truth relevance score $r^*_d$ for each item $d$ is not available. 

\subsection{Connections to Related Work}
\label{sec:related}
LTR algorithms refer to a group of machine techniques designed to solve ranking problems that have been successfully applied in multiple areas including top-N recommendation \citep{duan2010empirical}, 
search engine \citep{joachims2002optimizing,liu2011learning}. 
These techniques can be broadly categorized as point-wise, pair-wise, and list-wise methods regarding the computation of training losses.
The point-wise methods \citep{friedman2001greedy} directly treat the ranking problems as classifications or regressions by taking one item as the input to predict its score. 
The pair-wise methods \citep{burges2005learning,joachims2006training} take a pair of items as the input and optimize their relative positions in the final ranked list.
The list-wise methods \citep{burges2010ranknet,cao2007learning,xia2008listwise} further extend the above methods by taking multiple items together and optimizing the ranking metrics.
In addition, there are recently proposed models to build multi-variant scoring functions for each item \citep{ai2018learning} or design permutation-invariant encoders \citep{pang2020setrank,lee2019set} or pointer network encoders \citep{bello2018seq2slate} for the input item set.
However, all the above methods following PBP fall into a paradigm of first individually scoring items and then putting them into correct positions through a non-differentiable sort operation \citep{xia2008listwise} or a surrogate loss based on permutation probability \citep{plackett1975analysis}.
In contrast, \textsf{STARank} directly learns to generate the permutation of the candidate items does not involve individually scoring items and the sort operation. 

\section{The STARank Framework}
\label{sec:framework}
\subsection{Overview}
In this section, we present the architecture design of $\mathtt{f}_\mathtt{set}(\cdot|\cdot)$ as follows.
We first describe our reader module (as shown in Figure~\ref{fig:overview}(b)) to generate the representations of the candidate items in the context of the browsed items in Section~\ref{sec:reader}.
Then, we introduce our arranger module (as shown in Figure~\ref{fig:overview}(c)) to produce the permutation of the given candidate items in Section~\ref{sec:writer}.
We also answer how to derive a our list-wise loss $\ell_\mathtt{list}(\mathtt{f}_\mathtt{set}(\cdot|\cdot),\pi^*)$ (as shown in Figure~\ref{fig:overview}(d)) in Section~\ref{sec:supervision}.

\subsection{Reader Module}
\label{sec:reader}
For each user/query $q$ there are two types of inputs: 
(i) $\mathcal{D}_q$ is a set of candidate items, which is expected to be encoded in a permutation-invariant manner.
In other words, any permutation of the input items should not change the output ranking.
(ii) $\mathcal{H}_q$ is a sequence of browsed items that have their natural orders and thus should be encoded in a permutation-sensitive fashion.
Namely, the positions of browsed items should be considered.
In the light of this, we construct $\mathtt{f}_\mathtt{set}(\mathcal{D}_q|\mathcal{H}_q)$ as follows. 

Let $\bm{x}_d$ denote the feature vector of item $d$.
We first introduce a permutation-sensitive (PS) module, denoted as $\mathtt{PS}(\cdot)$, to represent the user embedding in a recursive manner as 
\begin{equation}
\label{eqn:user}
\bm{u}_q = \mathtt{PS}\Big(\bm{u}_0,\bm{x}_1,\ldots,\bm{x}_{|\mathcal{H}_q|}\Big),
\end{equation}
where $\bm{u}_q$ is the embedding vector for user $q$.
For simplicity, we regard the user profiles (e.g., gender, age) as the original user feature, denoted as $\bm{u}_0$ which is fed into our PS module as shown in Figure~\ref{fig:overview}(b).
While practical, we adopt a recurrent neural network (RNN), i.e., an LSTM \citep{hochreiter1997long} to form the PS module.

Given $\bm{u}_q$, we then deploy a permutation-invariant (PI) module, denoted as $\mathtt{PI}(\cdot)$ to encode the candidate items in $\mathcal{D}_q$, which can be expressed as 
\begin{equation}
\label{eqn:reader}
\{\bm{h}_d|d\in\mathcal{D}_q\} = \mathtt{PI}\Big(\{\bm{x}_d|d\in\mathcal{D}_q\},\bm{u}_q\Big),
\end{equation}
where $\bm{h}_d$ is the representation vector of item $d$.
In particular, our PI module is based on an attention mechanism \citep{vaswani2017attention}, where the representation vector of each item $d$ is produced by
\begin{equation}
\label{eqn:h'}
\begin{aligned}
\bm{h}'_d &= \bm{w}_1^\top \mathtt{tanh}(\bm{W}_1 \bm{x}_d +\bm{b}_1), \ d \in \mathcal{D}_q,
\end{aligned}
\end{equation}
\begin{equation}
\label{eqn:h}
\begin{aligned}
\bm{h}_d &= \beta_d \bm{h}'_d, \text{ where }
\beta_d = \frac{\bm{h}'_d \bm{u}_q^\top}{\sum_{j\in\mathcal{D}_q}\bm{h}'_j \bm{u}_q^\top},
\end{aligned}
\end{equation}
where $\bm{w}_\cdot$, $\bm{W}_\cdot$, $\bm{b}_\cdot$ are trainable matrices and vectors.

\begin{figure*}[t]
	\centering
	\includegraphics[width=1.0\linewidth]{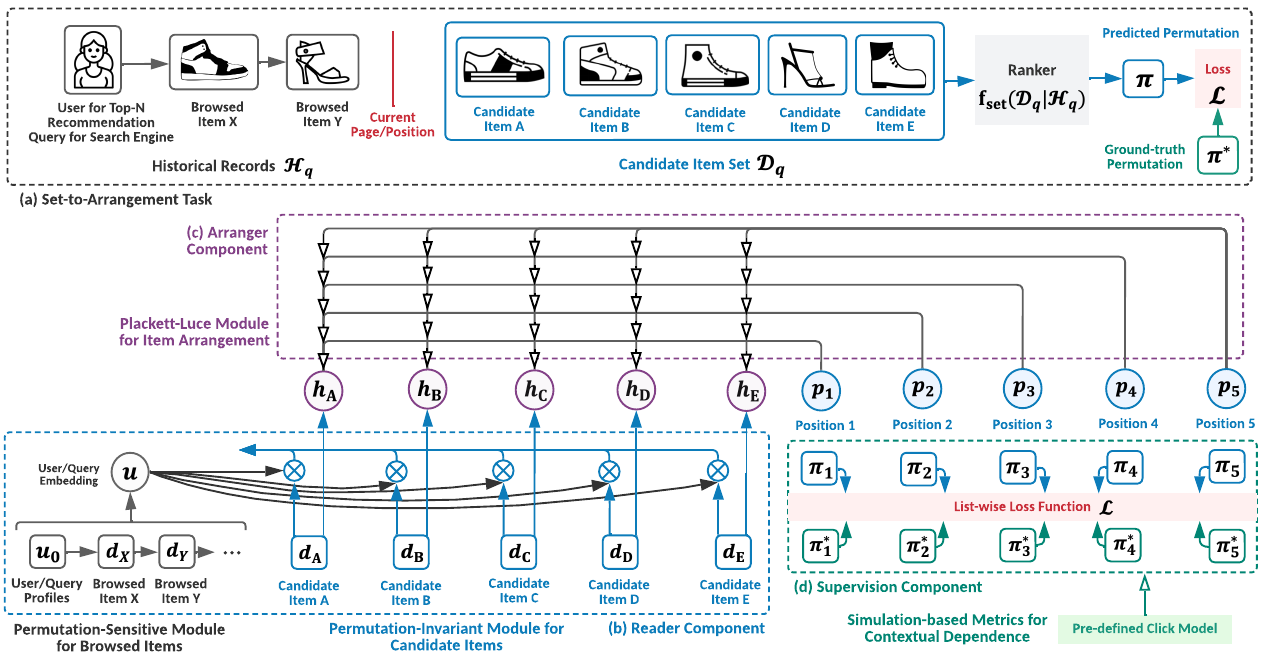}
	\vspace{-6mm}
	\caption {An illustrated example of \textsf{STARank}:
	(a) the task of \textsf{STARank} is to train a ranker $\mathtt{f}_\mathtt{set}(\mathcal{D}_q|\mathcal{H}_q)$ which produces the permutation $\pi$ of the candidate items $\mathcal{D}_q=\{d_\text{A},d_\text{B},d_\text{C},d_\text{D},d_\text{E}\}$ in the context of the user $q$ and the browsed items $\mathcal{H}_q=[d_\text{X},d_\text{Y}]$.
	Specifically, (b) \textsf{STARank} reads a set of candidate items into $\{\bm{h}_\text{A},\bm{h}_\text{B},\bm{h}_\text{C},\bm{h}_\text{D},\bm{h}_\text{E}\}$, where the browsed items together with the user profile are first encoded in a permutation-sensitive fashion, and then are fed into a permutation-invariant module to embed the candidate items.
	(c) \textsf{STARank} arranges the candidate items into their positions (namely generating the predicted permutation $\pi$) according to the attention scores of placing the item into the position that are learned via a Plackett-Luce module.
	(d) \textsf{STARank} is supervised by optimizing a list-wise loss function $\mathcal{L}$ over $\pi$ and the oracle permutation $\pi^*$ where $\pi^*$ is given.
    Additionally, we introduce a family of simulation-based ranking metrics.
	}
	\label{fig:overview}
\end{figure*}

By incorporating Eqs.~(\ref{eqn:user}), (\ref{eqn:h'}) and (\ref{eqn:h}) into Eq.~(\ref{eqn:reader}), we can obtain the representation vectors of items, being aware of the representations of the user and her browsed items. 

\subsection{Arranger Module}
\label{sec:writer}
With the representation vector of each input item, we use a Plackett-Luce (PL) module to generate their permutation.
Formally, the PL module, denoted as $\mathtt{PL}(\cdot)$, is defined as
\begin{equation}
\label{eqn:rho}
\pi = \mathtt{PL}\Big(\{\bm{h}_d|d\in\mathcal{D}_q\},\bm{u}_q\Big).
\end{equation}
Noticing that the original pointer network \citep{vinyals2015pointer} using LSTM as the encoder is permutation-sensitive, we tweak the pointer network by replacing the LSTM with the aforementioned reader.

More concretely, we construct a batch of one-hot position embedding vectors $\{\bm{p}_i\}_{i=1}^{|\mathcal{D}_q|}$ to encode the position information.
These embeddings can be obtained by applying a LSTM as the decoder.

Let $\pi_{i}$ denote the item at the $i$-th position, and $\pi_{<i}$ denote a set of the items whose positions are smaller than $i$, namely $\pi_{<i}=\{\pi_{1},\ldots,\pi_{i-1}\}$.
For convenience, we define $\pi_{<1}$ as $\emptyset$.

Then, the probability distribution of item $d$ at the $i$-th position can be produced as
\begin{equation}
\label{eqn:score}
s^i_{d}= \bm{v}^i_d \bm{u}_q^\top, \text{ where }\bm{v}^i_{d}=\bm{w}_{2}^\top \mathtt{tanh}(\bm{W}_2 \bm{h}_d + \bm{W}_3 \bm{p}_i + \bm{b}_{2}),
\end{equation}
\begin{equation}
\label{eqn:generate}
P(\pi_{i}=d|\pi_{<i})=  \left\{
	\begin{aligned}
	&e^{s^i_{d}}/\sum_{k\in \mathcal{D}_q\backslash\pi_{<i}} e^{s^i_{k}},  \text{ if } d \not\in \pi_{<i},\\
	&0, \text{ otherwise}.
	\end{aligned}
	\right.,
\end{equation}
where $i =1,\ldots,|\mathcal{D}_q|, d\in\mathcal{D}_q$;
$s^i_d$ is the trainable attention score associated with placing item $d$ at position $i$ and $p^i_d= P(\pi_{i}=d|\pi_{<i})$ is the probability representing the degree to which the PL module points item $d$ to  position $i$. 
$p^i_d$ is computed by a softmax over the remaining items and is set to 0 for items already been arranged.
Once item $d$ is placed at position $i$, its embedding $\bm{h}_d$ is fed as the input at the next step to produce the scores of position $i+1$.

We demonstrate the arrangement process in Figure~\ref{fig:overview}(c).
This way can ensure the module holds the information on the items already been arranged at each position to eventually output the predicted permutation $\pi$.

As Eq.~(\ref{eqn:generate}) shows the step-by-step generation of $\pi$, we follow Eq.~(\ref{eqn:p}) in \citep{hunter2004mm} to derive the predicted permutation $\pi$ as
\begin{equation}
\label{eqn:p}
P(\pi|\mathcal{D}_q) = \prod^{|\mathcal{D}_q|}_{i=1} \frac{e^{s^i_{\pi_i}}}{\sum^{|\mathcal{D}_q|}_{j=i}e^{s^j_{\pi_{j}}}},
\end{equation}
where $i=1,\ldots,|\mathcal{D}_q|$ and $s^i_{\pi_{i}}$ is computed by Eq.~(\ref{eqn:score}).
We apply Eq.~(\ref{eqn:p}) as $\mathtt{PL}(\cdot)$ in Eq.~(\ref{eqn:rho}).

\subsection{List-wise Loss Function}
\label{sec:supervision}
For each user/query $q$, we can compute the predicted permutation $\pi$, then one direct way to construct $\ell_\mathtt{list}(\pi,\pi^*)$ can be formulated as
$\ell_\mathtt{list}(\pi,\pi^*) = \sum^{|\mathcal{D}_q|}_{i=1}\ell_\mathtt{point}(\pi^*_{i},\pi_{i})$.
However, this formulation is a simple summation of point-wise loss that fails to consider the contextual dependence among items.

Notice that the generation of $\pi_i$ is conditioned on the arranged items (i.e., $\pi_{<i}$), as shown in Eq.~(\ref{eqn:generate}).
Therefore, the corresponding supervision of each $\pi_i$ should be $\pi^*_{i}|\pi_{<i}$, where $\pi^*|\pi_{<i}$ is the ground-truth permutation given the items in $\pi_{<i}$ have been arranged, and $\pi^*_{i}|\pi_{<i}$ is the item at position $i$ in $\pi^*|\pi_{<i}$.
In this regard, $\ell_\mathtt{list}(\pi,\pi^*)$ can be written as
\begin{equation}
\ell_\mathtt{list}(\pi, \pi^*) = \sum^{|\mathcal{D}_q|}_{i=1}\ell_\mathtt{point}(\pi^*_{i}|\pi_{<i},\pi_{i}|\pi_{<i}).
\end{equation}
However, $\pi^*|\pi_{<i}$ is not available, since we only have the ground-truth permutation $\pi^*$.
Alternatively, we construct the following surrogate loss to optimize the arrangement of position $i$ when the previous items are arranged following the ground-truth.
In this regard, $\ell_\mathtt{list}(\pi, \pi^*)$ can be written as
\begin{equation}
\label{eqn:lossseq}
\ell_\mathtt{list}(\pi, \pi^*) = \sum^{|\mathcal{D}_q|}_{i=1}\ell_\mathtt{point}(\pi^*_{i}|\pi^*_{<i},\pi_{i}|\pi^*_{<i}).
\end{equation}
However, calculating $\pi_{i}$ conditioned on $\pi^*_{<i}$ is still challenging.
Therefore, we leverage the internal consistency property of our arranger (i.e., PL models) to derive a computable version of $\ell_\mathtt{list}(\pi, \pi^*)$.

We derive the difference between $P(\pi_i|\pi^*_{<i})$ and $P(\pi^*_i|\pi^*_{<i})$ as follows.
\begin{equation}
\label{eqn:proof1}
\begin{aligned}
P(\pi_i|\pi^*_{<i}) = & \frac{P(\{\pi_i\}\cup\pi^*_{<i})}{P(\pi^*_{<i})}
=\frac{P(\{\pi_i\}\cup\{\pi^*_{i-1}\}\cup\pi^*_{<i-1})}{P(\pi^*_{<i})}\\
= &P(\{\pi_i\}\cup\{\pi^*_{i-1}\}|\pi^*_{<i-1})\cdot \frac{P(\pi^*_{<i-1})}{P(\pi^*_{<i})},
\end{aligned}
\end{equation}
where $P(\pi_i|\pi^*_{<i})$ denotes the probability of arranging $\pi_i$ at position $i$ conditioned on arranged items $\pi^*_{<i}$, and $P(\{\pi_i\}\cup\pi^*_{<i})$ denotes the probability of arranging $\pi^*_{<i}$ at positions from 1 to $i-1$ and arranging $\pi_i$ at position $i$. 
Eq.~(\ref{eqn:proof1}) tells us the main difference lies in $P(\{\pi_i\}\cup\{\pi^*_{i-1}\}|\pi^*_{<i-1})$ and $P(\{\pi^*_i\}\cup\{\pi^*_{i-1}\}|\pi^*_{<i-1})$.

We then formally describe the internal consistency property of PL models in the following lemma.
\begin{lemma}
\label{lemma:consistent}
Given a set of items $\mathcal{D}_q$ and one subset $\mathcal{D}'_q\subseteq \mathcal{D}_q$, the internal permutation of the items in $\mathcal{D}'_q$, denoted as $\pi'$, is consistent in the context of either $\mathcal{D}_q$ or $\mathcal{D}'_q$, which can be formulated as 
\begin{equation}
\label{eqn:lemma}
P(\pi'|\mathcal{D}_q) = P(\pi'|\mathcal{D}'_q),
\end{equation}
where $P(\pi'|\mathcal{D}_q)$ and $P(\pi'|\mathcal{D}'_q)$ are the probabilities of the items in $\mathcal{D}'_q$ satisfying $\pi'$ in the context of using the candidate item sets $\mathcal{D}_q$ and $\mathcal{D}'_q$ as inputs respectively.
\end{lemma}
This property is originally proposed in \citep{hunter2004mm}, and  one can refer to Eq.~(27) in \citep{hunter2004mm} for the derivation of \textsc{Lemma~\ref{lemma:consistent}} (i.e., Eq.~(\ref{eqn:lemma})).

Then, by assigning $\mathcal{D}'_q=\{\pi_i,\pi^*_{i-1}\}$ and $\mathcal{D}_q=\{\pi_i\}\cup\pi^*_{<i}$, we can derive that $P(\{\pi_i\}\cup\{\pi^*_{i-1}\})=P(\{\pi_i\}\cup\{\pi^*_{i-1}\}|\pi^*_{<i-1})$, namely the relative permutation of $\pi_i$ and $\pi^*_{i-1}$ is independent of how positions from 1 to $i-2$ are arranged.
In other words, the main difference between $P(\pi_i|\pi^*_{<i})$ and $P(\pi^*_i|\pi^*_{<i})$, as derived in Eq.~(\ref{eqn:proof1}) exists in $P(\pi_i|\pi^*_{i-1})$ and $P(\pi^*_i|\pi^*_{i-1})$.

For convenience, let $\pi^*_0=\emptyset$.
Then, we re-write Eq.~(\ref{eqn:lossseq}) as
\begin{equation}
\label{eqn:proof2}
\ell_\mathtt{list}(\pi, \pi^*) = \sum^{|\mathcal{D}_q|}_{i=1}\ell_\mathtt{point}(\pi_{i}|\pi^*_{i-1},\pi^*_{i}|\pi^*_{i-1}).
\end{equation}
Eq.~(\ref{eqn:proof2}) establishes a list-wise approach to optimize \textsf{STARank} by supervising the arrangement at each position.
In practice, we formulate the overall loss as
\begin{equation}
\label{eqn:allloss}
\mathcal{L}=- \sum_q\sum^{|\mathcal{D}_q|}_{i=1} \log P(\pi_i=\pi^*_{i}|\pi^*_{i-1}),
\end{equation}
where $P(\pi_i=\pi^*_{i}|\pi^*_{i-1})$ is the probability of arranging $\pi^*_{i}$ after the arranged item $\pi^*_{i-1}$.

Compared to existing loss based on the ground-truth relevance scores, $\mathcal{L}$ pushes \textsf{STARank} to fit the ground-truth permutation directly in a position-by-position way, which corresponds to the recursive generation process in our arranger (as shown in Eq.~(\ref{eqn:generate})). 
In other words, $\mathcal{L}$ focuses on identifying and correcting the relative permutations of neighboring pairs that have errors during generating the arrangement.

\begin{algorithm}[t]
	\caption{\textsf{STARank}}
	\label{algo:framework}
	\begin{algorithmic}[1]
		\REQUIRE
		item set $\mathcal{D}_q$, historical data $\mathcal{H}_q$, ground-truth $\pi^*$.
		\ENSURE
		ranker $\mathtt{f}_\mathtt{set}(\cdot|\cdot)$, predicted permutation $\pi$.
		\vspace{1mm}
		\STATE Initialize all parameters.
		\REPEAT
		\FOR {each user (or query) $q$}
		\STATE Read $\mathcal{H}_q$, $\mathcal{D}_q$ into $\bm{u}_q$, $\{\bm{h}_d|d\in\mathcal{D}_q\}$ via Eqs.~(\ref{eqn:user}) and (\ref{eqn:reader}).
		\STATE Recursively arrange $\mathcal{D}_q$ into $\pi$  using Eq.~(\ref{eqn:rho}).
		\ENDFOR
		\STATE Update parameters by minimizing $\mathcal{L}$ using Eq.~(\ref{eqn:allloss}).
		\UNTIL convergence
	\end{algorithmic}
\end{algorithm}

\subsection{Model Analysis} 
\label{sec:analysis}

\minisection{Complexity.}
From Algorithm~\ref{algo:framework}, we can see that the main components of \textsf{STARank} are the reader component and the arranger component.

Let $L$ denote the length of candidate items of user/query $q$, then the time complexity of computing an attention score of placing an item is $O(LF_1F_2+EF_1)$, where $E$ is the number of input features, and $F_1$ and $F_2$ are the number of rows and columns of the attention matrix respectively.
The above process needs to iterate for each item and for each user/query.
Therefore, the overall complexity is $O(QL^2F_1F_2+QLEF_1)$ where $Q$ is the number of users/queries.

\minisection{Property.}
We theorize that \textsf{STARank} is permutation-invariant regarding the candidate items in the following lemma:
\begin{lemma}
\label{pro:analysis}
For the given item space $\mathcal{D}$, assume that its elements are finite, i.e., $|\mathcal{D}|<\infty$.
For each candidate item set $\mathcal{D}_q\subseteq \mathcal{D}$ and its corresponding browsed items $\mathcal{H}_q\subseteq \mathcal{D}$, the ranker $\mathtt{f}_\mathtt{set}(\mathcal{D}_q|\mathcal{H}_q)$ in Eq.~(\ref{eqn:riskset}) following \emph{\textsf{STARank}} is permutation-invariant to the elements in $\mathcal{D}_q$.
\end{lemma}
\begin{proof}
We first re-express \textsf{STARank} as a function producing the predicted scores for all the possible permutations, namely $\mathtt{f}_\mathtt{set}:\mathcal{P}\rightarrow \mathbb{R}$ which can be specified as
\begin{equation}
P(\pi|\mathcal{D}_q,\mathcal{H}_q) = \mathtt{PL}\left(\mathtt{PI}(\mathcal{D}_q, \mathtt{PS}(\mathcal{H}_q)) \right).
\end{equation}
Here, as $\mathcal{H}_q$ is independent of $\mathcal{D}_q$, we can omit $\mathtt{PS}(\mathcal{H}_q)$ here.
As shown in Eq.~(\ref{eqn:h}), $\mathtt{PI}(\mathcal{D}_q)$ employing an attention mechanism can be regarded as the summarization of the items in $\mathcal{D}_q$ weighted by the attention scores.
And, the formulation of $\mathtt{PL}(\cdot)$ is provided in Eq.~(\ref{eqn:p}).
According to \textsc{Theorem}~2 in \citep{zaheer2017deep}, one can easily conclude that the whole \textsf{STARank} framework is permutation-invariant regarding the candidate items.
\end{proof}

\begin{table*}[t]
	\centering
	\caption{Comparison of different rankers on three industrial top-N recommendation datasets in terms of NDCG and MAP at positions 5, 10. 
	* indicates $p < 0.001$ in significance tests compared to the best baseline.
	}
	\vspace{-3mm}
	\resizebox{1.00\textwidth}{!}{
		\begin{tabular}{@{\extracolsep{4pt}}|c|cccc|cccc|cccc|}
			\toprule
			\multirow{2}{*}{Ranker} & \multicolumn{4}{c}{Tmall} |& \multicolumn{4}{c}{Alipay} |& \multicolumn{4}{c}{Taobao} |\\
			\cmidrule{2-5}
			\cmidrule{6-9}
			\cmidrule{10-13}
			{} & N@5 & N@10 & M@5 & M@10 & N@5 & N@10 & M@5 & M@10 & N@5 & N@10 & M@5 & M@10 \\
			\midrule
			FM & 
			0.1233 & 0.1140 & 0.2494 & 0.2572 & 
			0.1342 & 0.1277 & 0.2634 & 0.2709 & 
			0.1124 & 0.1106 & 0.2228 & 0.2354 \\
			DeepFM & 
			0.1241 & 0.1181 & 0.2515 & 0.2591 & 0.1276 & 0.1214 & 0.2551 & 0.2629 & 0.1117 & 0.1104 & 0.2209 & 0.2339 \\
			PNN & 
			0.1248 & 0.1185 & 0.2525 & 0.2602 & 0.1296 & 0.1227 & 0.2601 & 0.2666 & 0.1145 & 0.1123 & 0.2266 & 0.2385 \\
			LSTM & 
			0.1341 & 0.1287 & 0.2660 & 0.2716 & 0.1427 & 0.1348 & 0.2783 & 0.2816 & 0.1308 & 0.1284 & 0.2496 & 0.2587 \\
			GRU & 
			0.1311 & 0.1255 & 0.2606 & 0.2668 & 0.1422 & 0.1345 & 0.2773 & 0.2810 & 0.1300 & 0.1275 & 0.2488 & 0.2578 \\
			DIN & 
			0.1345 & 0.1305 & 0.2674 & 0.2725 & 0.1401 & 0.1365 & 0.2705 & 0.2751 & 0.1230 & 0.1227 & 0.2351 & 0.2457 \\
			DIEN & 
			0.1243 & 0.1182 & 0.2518 & 0.2594 & 0.1351 & 0.1293 & 0.2632 & 0.2696 & 0.1132 & 0.1122 & 0.2245 & 0.2366 \\
			SetRank & 
			0.2733 & 0.2640 & 0.4090 & 0.4284 & 
			0.2969 & 0.2850 & 0.4289 & 0.4458 & 
			0.2594 & 0.2498 & 0.3876 & 0.4062 \\
            Seq2Slate & 
			0.2385 & 0.2340 & 0.3720 & 0.3948 & 
			0.2452 & 0.2412 & 0.3771 & 0.3993 &
			0.2223 & 0.2180 & 0.3437 & 0.3690 \\
			\midrule
			\textsf{STARank}$^-_\mathtt{PI}$ & 
			0.3042 & 0.3251 & 0.4095 & 0.4254 & 
			0.3214 & 0.3264 & 0.4014 & 0.4446 & 
			0.2632 & 0.2734 & 0.3974 & 0.4214 \\
			\textsf{STARank}$^-_\mathtt{PS}$ & 
			0.2952 & 0.3052 & 0.3974 & 0.4035 & 
			0.3046 & 0.3035 & 0.4046 & 0.4256 & 
			0.2678 & 0.2742 & 0.4025 & 0.4264 \\
			\textbf{\textsf{STARank}} & 
			\textbf{0.3353}$^*$ & \textbf{0.3745}$^*$ & \textbf{0.4328}$^*$ & \textbf{0.4358}$^*$ & \textbf{0.4509}$^*$ & \textbf{0.4803}$^*$ & \textbf{0.5337}$^*$ & \textbf{0.5703}$^*$ & \textbf{0.3479}$^*$ & \textbf{0.4278}$^*$ & \textbf{0.4082}$^*$ & \textbf{0.4635}$^*$ \\
			\textbf{\textsf{STARank}$^+_\mathtt{PBM}$} & 
			\textbf{0.3591}$^*$ & \textbf{0.4311}$^*$ & \textbf{0.4472}$^*$ & \textbf{0.4932}$^*$ & \textbf{0.4725}$^*$ & \textbf{0.5046}$^*$ & \textbf{0.5402}$^*$ & \textbf{0.5854}$^*$ & \textbf{0.4006}$^*$ & \textbf{0.4868}$^*$ & \textbf{0.4683}$^*$ & \textbf{0.5250}$^*$ \\
			\textbf{\textsf{STARank}$^+_\mathtt{UBM}$} & 
			\textbf{0.3482}$^*$ & \textbf{0.4208}$^*$ & \textbf{0.4359}$^*$ & \textbf{0.4828}$^*$ & \textbf{0.4481}$^*$ & \textbf{0.4709}$^*$ & \textbf{0.5189}$^*$ & \textbf{0.5644}$^*$ & \textbf{0.3930}$^*$ & \textbf{0.4792}$^*$ & \textbf{0.4545}$^*$ & \textbf{0.5164}$^*$ \\
			\bottomrule
		\end{tabular}
	}
	\label{tab:res}
	\vspace{-2mm}
\end{table*}

\begin{table}[t]
	\centering
	\caption{Comparison of different rankers on  Yahoo dataset.
    * indicates $p < 0.001$ in significance tests.
	}
	\vspace{-3mm}
	\resizebox{0.85\linewidth}{!}{
		\begin{tabular}{@{\extracolsep{4pt}}|c|cccc|}
			\toprule
			\multirow{2}{*}{Ranker} & \multicolumn{4}{c}{Yahoo}|\\
			\cmidrule{2-5}
			{} & N@5 & N@10 & M@5 & M@10 \\
			\midrule
			FM & 
			0.2122 & 0.2434 & 0.3574 & 0.3629 \\
			DeepFM & 
			0.2483 & 0.2490 & 0.3587 & 0.3662 \\
			PNN & 
			0.2563 & 0.2693 & 0.3598 & 0.3602 \\
			LSTM & 
			0.2763 & 0.2731 & 0.3602 & 0.3641 \\
			GRU & 
			0.2901 & 0.2876. & 0.3674 & 0.3675 \\
			DIN & 
			0.2932 & 0.2975. & 0.3772 & 0.3706 \\
			DIEN & 
			0.2911 & 0.2983 & 0.3703 & 0.3688 \\
			SetRank & 
			0.3053 & 0.3185 & 0.3932 & 0.3868 \\
            Seq2Slate & 
			0.3121 & 0.3234 & 0.4001 & 0.3987 \\
			\midrule
			\textsf{STARank}$^-_\mathtt{PI}$ & 
			0.3188 & 0.3210 & 0.4079 & 0.4154 \\
			\textsf{STARank}$^-_\mathtt{PS}$ & 
			0.3110 & 0.3356 & 0.4139 & 0.4201 \\
			\textbf{\textsf{STARank}} & 
			\textbf{0.3399}$^*$ & \textbf{0.3525}$^*$ & \textbf{0.4301}$^*$ & \textbf{0.4398}$^*$ \\
			\textbf{\textsf{STARank}$^+_\mathtt{PBM}$} & 
			\textbf{0.3474}$^*$ & \textbf{0.3612}$^*$ & \textbf{0.4342}$^*$ & \textbf{0.4450}$^*$ \\
			\textbf{\textsf{STARank}$^+_\mathtt{UBM}$} & 
			\textbf{0.3355}$^*$ & \textbf{0.3543}$^*$ & \textbf{0.4324}$^*$ & \textbf{0.4463}$^*$ \\
			\bottomrule
		\end{tabular}
		}
	\vspace{-2mm}
	\label{tab:rank}
\end{table}

\begin{table}[t]
	\centering
	\caption{Comparison of different rankers on LETOR dataset.
    * indicates $p < 0.001$ in significance tests.
	}
	\vspace{-3mm}
	\resizebox{0.85\linewidth}{!}{
		\begin{tabular}{@{\extracolsep{4pt}}|c|cccc|}
			\toprule
			\multirow{2}{*}{Ranker} & \multicolumn{4}{c}{LETOR}|\\
			\cmidrule{2-5}
			{} & N@5 & N@10 & M@5 & M@10 \\
			\midrule
			FM & 
			0.1752 & 0.1721 & 0.3282 & 0.3292 \\
			DeepFM & 
			0.1746 & 0.1689 & 0.3243 & 0.3204 \\
			PNN & 
			0.1775 & 0.1739 & 0.3296 & 0.3301 \\
			LSTM & 
			0.1851 & 0.1795 & 0.3405 & 0.3384 \\
			GRU & 
			0.1739 & 0.1741. & 0.3170 & 0.3205 \\
			DIN & 
			0.1550 & 0.1671. & 0.2815 & 0.2988 \\
			DIEN & 
			0.1818 & 0.1811 & 0.3409 & 0.3402 \\
			SetRank & 
			0.2035 & 0.2246 & 0.4045 & 0.4002 \\
            Seq2Slate & 
			0.2436 & 0.2546 & 0.4122 & 0.4031 \\
			\midrule
			\textsf{STARank}$^-_\mathtt{PI}$ & 
			0.2443 & 0.2568 & 0.4172 & 0.4075 \\
			\textsf{STARank}$^-_\mathtt{PS}$ & 
			0.2467 & 0.2606 & 0.4231 & 0.4123 \\
			\textbf{\textsf{STARank}} & 
			\textbf{0.2824}$^*$ & \textbf{0.2622}$^*$ & \textbf{0.4876}$^*$ & \textbf{0.4529}$^*$ \\
			\textbf{\textsf{STARank}$^+_\mathtt{PBM}$} & 
			\textbf{0.2843}$^*$ & \textbf{0.2646}$^*$ & \textbf{0.4906}$^*$ & \textbf{0.4573}$^*$ \\
			\textbf{\textsf{STARank}$^+_\mathtt{UBM}$} & 
			\textbf{0.2863}$^*$ & \textbf{0.2668}$^*$ & \textbf{0.4887}$^*$ & \textbf{0.4587}$^*$ \\
			\bottomrule
		\end{tabular}
		}
	\vspace{-2mm}
	\label{tab:leto}
\end{table}

\begin{table*}[t]
	\centering
	\caption{Comparison of different rankers on three industrial top-N recommendation datasets in terms of the proposed simulation-based ranking metrics based on PBM and UBM at positions 5, 10.
	* indicates $p < 0.001$ in significance tests compared to the best baseline.
	}
	\vspace{-3mm}
	\label{tab:seq}
	\resizebox{1.00\textwidth}{!}{
		\begin{tabular}{@{\extracolsep{4pt}}|c|cccc|cccc|cccc|}
			\toprule
			\multirow{2}{*}{Ranker} & \multicolumn{4}{c}{Tmall} |& \multicolumn{4}{c}{Alipay} |& \multicolumn{4}{c}{Taobao} |\\
			\cmidrule{2-5}
			\cmidrule{6-9}
			\cmidrule{10-13}
			{} & $\mathtt{P}$@5 & $\mathtt{P}$@10 & $\mathtt{U}$@5 & $\mathtt{U}$@10 & $\mathtt{P}$@5 & $\mathtt{P}$@10 & $\mathtt{U}$@5 & $\mathtt{U}$@10 & $\mathtt{P}$@5 & $\mathtt{P}$@10 & $\mathtt{U}$@5 & $\mathtt{U}$@10 \\
			\midrule
			FM & 
			0.1026 & 0.1027 & 0.1015 & 0.1016 & 0.1138 & 0.1157 & 0.1013 & 0.1014 & 0.1073 & 0.1067 & 0.1013 & 0.1017 \\
			DeepFM & 
			0.1023 & 0.1023 & 0.1005 & 0.1008 & 0.1037 & 0.1040 & 0.1008 & 0.1009 & 0.1093 & 0.1097 & 0.1008 & 0.1015 \\
			PNN & 
			0.1248 & 0.1185 & 0.1014 & 0.1016 & o.1023 & 0.1021 & 0.1002 & 0.1005 & 0.1142 & 0.1155 & 0.101 & 0.1016 \\
			LSTM & 
			0.2074 & 0.2734 & 0.1301 & 0.1460 & 0.2137 & 0.2068 & 0.1274 & 0.1488 & 0.2498 & 0.2922 & 0.2025 & 0.2842 \\
			GRU & 
			0.2188 & 0.3066 & 0.1299 & 0.1437 & 0.2292 & 0.2208 & 0.1308 & 0.1512 & 0.2541 & 0.3023 & 0.1928 & 0.2747 \\
			DIN & 
			0.2725 & 0.2474 & 0.2493 & 0.1818 & 0.2257 & 0.2333 & 0.1443 & 0.1648 & 0.2748 & 0.2741 & 0.2099 & 0.2846 \\
			DIEN & 
			0.2520 & 0.3693 & 0.1014 & 0.1012 & 0.2366 & 0.2194 & 0.1281 & 0.1456 & 0.2345 & 0.2986 & 0.1891 & 0.2654 \\
			SetRank & 
			0.3778 & 0.3973 & 0.3378 & 0.3650 & 
			0.4171 & 0.4330 & 0.3346 & 0.3584 & 
			0.3094 & 0.3191 & 0.3704 & 0.3820 \\
            Seq2Slate & 
			0.3182 & 0.3455 & 0.3210 & 0.3482 & 
			0.3337 & 0.3600 & 0.3254 & 0.3533 &
			0.3285 & 0.3541 & 0.3189 & 0.3470 \\
			\midrule
			\textsf{STARank}$^-_\mathtt{PI}$ & 
			0.3962 & 0.3967 & 0.4346 & 0.4426 & 
			0.5025 & 0.5046 & 0.4325 & 0.4602 & 
			0.3429 & 0.3478 & 0.4385 & 0.4255 \\
			\textsf{STARank}$^-_\mathtt{PS}$ & 
			0.3873 & 0.3863 & 0.4253 & 0.4371 & 
			0.5010 & 0.5024 & 0.4242 & 0.4523 & 
			0.3321 & 0.3368 & 0.4257 & 0.4211 \\
			\textbf{\textsf{STARank}} & 
			\textbf{0.4240}$^*$ & \textbf{0.4010}$^*$ & \textbf{0.4435}$^*$ & \textbf{0.4493}$^*$ & \textbf{0.5283}$^*$ & \textbf{0.5058}$^*$ & \textbf{0.4578}$^*$ & \textbf{0.4888}$^*$ & \textbf{0.3529}$^*$ & \textbf{0.3599}$^*$ & \textbf{0.4418}$^*$ & \textbf{0.4308}$^*$ \\
			\bottomrule
		\end{tabular}
	}
\end{table*}

\section{Experiment}
\label{sec:exp}

\subsection{Datasets and Data Pre-processing}

\minisection{Dataset.}
We introduce 3 industrial recommendation datasets, namely Tmall, Taobao, and Alipay, and 2 learning-to-rank benchmark datasets, namely Yahoo and LETOR in our experiment. 
\begin{itemize}[topsep = 3pt,leftmargin =10pt]
\item \textbf{Tmall}\footnote{\url{https://tianchi.aliyun.com/dataset/dataDetail?dataId=42}} is a dataset consisting of 54,925,331 interactions of 424,170 users and 1,090,390 items.
These sequential histories are collected by the Tmall e-commerce platform from May 2015 to November 2015 with an average sequence length of 129 and 9 feature fields. 
\item \textbf{Taobao}\footnote{\url{https://tianchi.aliyun.com/dataset/dataDetail?dataId=649}} is a dataset containing 100,150,807 interactions of 987,994 users and 4,162,024 items.
These user behaviors including several behavior types (e.g., click, purchase, add to chart, item favoring) are collected from November 2007 to December 2007 with an average sequence length of 101 and 4 feature fields.
\item \textbf{Alipay}\footnote{\url{https://tianchi.aliyun.com/dataset/dataDetail?dataId=53}} is a dataset collected by Alipay, an online payment application from July 2015 to November 2015.
There are 35,179,371 interactions of 498,308 users and 2,200,191 items with an average sequence length of 70 and 6 feature fields.
\item \textbf{Yahoo}\footnote{\url{http://webscope.sandbox.yahoo.com}} is a dataset collected by Yahoo, an online search engine.
It consists of 29,921 queries and 710k items. 
Each query document pair is represented by a 700-dimensional feature vector manually assigned with a label denoting relevance at 5 levels \cite{chapelle2011yahoo}.
\item \textbf{LETOR}\footnote{\url{https://www.microsoft.com/en-us/research/project/letor-learning-rank-information-retrieval/}} is a package of benchmark datasets for research on LEarning TO Rank, which uses the Gov2 web page collection and two query sets from the Million Query track of TREC2007 and TREC2008.
We conduct experiments on MQ2007, one of two datsets for supervised learning-to-rank tasks in LETOR. MQ2007 contains 2,476 queries and 85K documents.
Each query-document pair is represented by a 46-dimensional feature vector with a manually assigned 3-level label \citep{qin2010letor}.
\end{itemize}

\minisection{Data Pre-processing}\footnote{We note that our setting on how to divide datasets is different from previous work such as SetRank \citep{pang2020setrank} which would make the results not consistent with their results reported in original papers, because our framework requires a sequence of items as the history (i.e., $\mathcal{H}_q$) and another set of items as the candidate items (i.e., $\mathcal{D}_q$).}.
We split each dataset using the timestep.
Let $T$ denote the length of user browsing logs.
We first filter out the data instances whose browsing logs are smaller than 30 (i.e., $T < 30$).
Next, for each user $q$, we assign 1-st to ($T$-20)-th items as the training dataset where 1-st to ($T-30$)-th items are used as the user's history data (i.e., $\mathcal{H}_q$) and ($T-29$)-th to ($T-20$)-th items construct the candidate item set (i.e., $\mathcal{D}_q$).
In validation, we use 1-st to ($T-20$)-th items as $\mathcal{H}_q$ and the ($T-19$)-th to ($T-10$)-th items as $\mathcal{D}_q$; while for evaluation, we use 1-st to ($T-10$)-th items as $\mathcal{H}_q$ and ($T-9$)-th to $T$-th items as $\mathcal{D}_q$.
For $\mathcal{H}_q$ and $\mathcal{D}_q$, we keep their original browsing orders and do not adopt any initial ranker to generate the initial ranking list.

\subsection{Experimental Configurations}
\minisection{Simulation-based Ranking Metrics.}
Consider that real-world datasets usually offer ground-truth relevance scores for candidate items, instead of ground-truth permutations. 
Therefore, we are required to first generate the ground-truth permutations.
For this purpose, we introduce a ranking metric, denoted as $\mathtt{R}(\cdot)$.
Let $\mathcal{P}$ denote the set of all the possible permutation of $\mathcal{D}_q$.
Then, we have $|\mathcal{P}|=\mathtt{A}^{|\mathcal{D}_q|}_{|\mathcal{D}_q|}=O(|\mathcal{D}_q|!)$ where $\mathtt{A}^\cdot_\cdot$ represents permutation operator.
After we enumerate all the possible permutations $\pi\in\mathcal{P}$, we select the one received the highest score in terms of $\mathtt{R}(\cdot)$ as the oracle permutation $\pi^*$.
In other words, we have: 
\begin{equation}
\label{eqn:oracle}
\pi^* = \text{arg} \mathop{\text{max}}_{\pi\in\mathcal{P}}\mathtt{R}(\pi).
\end{equation}
We note that conventional ranking metrics, such as NDCG, only consider the effects of the positions of the arranged items but overlook the contextual dependence among the items (i.e., the probability of a user favoring an item should be dependent on other items posted in the same sequence).

To this end, we propose to construct a virtual user to consider the contextual dependence and make the evaluation according to the user's feedbacks.
These virtual user models are known as the click models \citep{chuklin2015click}, which generate users' feedback by making some reasonable assumptions based on some heuristics or eye-tracking experimental results \cite{joachims2005accurately}.
Formally, we use $\mathtt{CM}$ to denote an arbitrary click model (CM) (e.g., position-based model (PBM) \citep{richardson2007predicting}, user browsing model (UBM) \citep{dupret2008user})
treated as a black block to generate the user's feedback for the input sequence.
Then, we introduce a new family of the ranking metrics based on the aforementioned simulations (i.e., click models), which is called simulation-based ranking metrics and defined as 
\begin{equation}
	\label{eqn:cm}
	\begin{aligned}
		\mathtt{R}_\mathtt{CM}(\pi)=  \sum^{|\mathcal{D}_q|}_{i=1}P(c_{\pi_i}|\pi_{<i};\mathtt{CM})
		=\sum^{|\mathcal{D}_q|}_{i=1} P(o_{\pi_i}|\pi_{<i};\mathtt{CM})\cdot P(r_{\pi_i}|\pi_{<i};\mathtt{CM}),
	\end{aligned}
\end{equation}
which reflects a simple fact that a user clicks an item $d$ (i.e., $P(c_d=1)$) only when it is both observed (i.e., $P(o_d=1)$) and perceived as relevant (i.e., $P(r_d=1)$),
and the observation probability and the relevance probability are determinated by the click model $\mathtt{CM}$.
We use PBM and UBM as examples for $\mathtt{CM}$ and introduce their detailed configurations as follows.
\textbf{PBM} \citep{richardson2007predicting} simulates user browsing behavior based on the assumption that the bias of a document only depends on its position, which can be formulated as $P(o_i) = \rho_i^\tau$, where $\rho_i$ represents 
position bias at position $i$ and $\tau \in [0, +\infty]$ is a parameter controlling the degree of position bias.
The position bias $\rho_i$ is obtained from an eye-tracking experiment in \citep{joachims2005accurately} and the parameter $\tau$ is set as 1 by default.
It also assumes that a user decides to click a document $d_i$ according to the probability $P(c_i) = P(o_i) \cdot P(r_i)$.
\textbf{UBM} \citep{dupret2008user} is an extension of the PBM model that has some elements of the cascade model.
The examination probability depends not only on the rank of an item $d_i$ but also on the rank of the previously clicked document $d_{i'}$ as $P(o_i=1|c_{i'}=1,c_{i'+1}=0,\ldots,c_{i-1}=0)=\upgamma_0$.
Similarly, we get $\upgamma_0$ from the eye-tracking experiments in \citep{joachims2005accurately,dupret2008user}. The click probability is determined by $P(c_i) = P(o_i) \cdot P(r_i)$. 

In this regard, NDCG metric is a special case of the proposed ranking metrics using PBM as $\mathtt{CM}$.
Concretely, the formulation of NDCG can be written as
\begin{equation}
\label{eqn:ndcg}
\mathtt{R}_\mathtt{NDCG}(\pi)=\frac{1}{N} \sum^{|\mathcal{D}_q|}_{i=1} (\frac{1}{\log_2{(i+1)}}) (2^{r_{\pi_i}}-1),
\end{equation}
where $N$ normalizes for the number of relevant items, and $r_{\pi_i}$ is the relevance of item $\pi_i$.
By comparing Eq.~(\ref{eqn:ndcg}) to Eq.~(\ref{eqn:cm}), we reveal that $\mathtt{R}_\mathtt{NDCG}(\cdot)$ is a special case of $\mathtt{R}_\mathtt{CM}(\cdot)$ by assigning
$P(o_{\pi_i}|\pi_{<i};\mathtt{CM})=\frac{1}{N_o}(\frac{1}{\log_2{(i+1)}}), P(r_{\pi_i}|\pi_{<i};\mathtt{CM})=\frac{1}{N_r}(2^{r_{\pi_i}}-1),$
where $N=N_o\cdot N_r$.
We can see that the observation probability of an arbitrary item is solely determined by its position, showing that CM here is a special case of PBM.

Noting that \textsf{STARank} metrics can be seamlessly adopted in the supervision generation, we can re-formulate Eq.~(\ref{eqn:oracle}) as
\begin{equation}
	\label{eqn:seqmetric}
	\pi^* = \text{arg} \mathop{\text{max}}_{\pi\in\mathcal{P}}\mathtt{R}_\mathtt{CM}(\pi).
\end{equation}

Note that the supervisions of \textsf{STARank} are determined by the ranking metrics.
For those datasets described in with binary scores in the browsing logs, there might exist multiple oracle permutations when applying either Eq.~(\ref{eqn:oracle}) or (\ref{eqn:seqmetric}).
In these cases, we randomly choose one oracle permutation as the supervision.
Also, these cases would occur during the evaluation phase, where we act similarly. 
We notice that there is a related paper \citep{sun2019re} investigating the impacts of different evaluation scheme regarding the same scores, which is out of the scope of this paper and we leave it for future work.


\minisection{Hyperparameter Setting.}
The learning rate is decreased from the initial value $1\times 10^{-2}$ to $1\times 10^{-6}$ during the training process.
The batch size is set as 100.
The weight for the L2 regularization term is $4\times 10^{-5}$.
The dimension of embedding vectors is set as 64.
The dropout rate is 0.5.
All experimental comparisons are conducted with 10 different random seeds.
All the models are trained under the same hardware settings with 6-Core AMD Ryzen 9 5950X (2.194GHZ), 62.78GB RAM, NVIDIA GeForce RTX 3080 cards.

\subsection{Baseline Methods and Ranking Metrics}
\label{sec:baseline}
\minisection{Baseline Descriptions.}
We compare \textsf{STARank} against 9 strong baselines, introduced as follows.
\textbf{FM} \citep{rendle2010factorization} is the factorization machine using the linear projection
and inner product of features to measure the user-item relevance.
\textbf{DeepFM} \citep{guo2017deepfm} is a generalized model consisting of a FM as a wide component and a MLP as a deep component.
\textbf{PNN} \citep{qu2016product} is product-based neural networks including an embedding layer and a product layer to model user-item interactions.
\textbf{DIN} \citep{zhou2018deep} designs a local activation unit to capture user interests
from historical records.
\textbf{LSTM} \citep{hochreiter1997long} is a standard long short memory approach widely
used for modeling user’s sequential pattern.
\textbf{GRU} \citep{hidasi2015session} uses the gate recurrent units to model sequential user
behaviors.
\textbf{DIEN} \citep{zhou2019deep} is an extension of DIN that builds an interest extractor
layer to model user’s temporal interests.
\textbf{SetRank} \citep{pang2020setrank} uses a stack of self-attention blocks based on set transformer \citep{lee2019set} to model the cross-item interactions.
\textbf{Seq2Slate} \citep{bello2018seq2slate} employs the pointer network \citep{qi2017pointnet} to sequentially
encode previously selected items and predict the next one.
All the methods are trained based on click feedback only.

\minisection{Evaluation Metrics.}
We evaluate the performance of these methods according to two series of ranking metrics, namely conventional ranking metrics and simulation-based ranking metrics built upon the specific click models.
To be more specific, for the former ones, we choose MAP (Mean Average Precision) and NDCG (Normalized Distributed Cumulative Gain) at position 5, 10, denoted as M@K and N@K where K is 5, 10.
For the later ones, we use PBM and UBM as the click models and compute the cumulative click probability at positions 5, 10, which are denoted as $\mathtt{P}$@K and $\mathtt{U}$@K and K is 5, 10.

Considering that the implementation of \textsf{STARank} is conditioned on which click model we choose.
Here, besides the standard implementation of \textsf{STARank} using NDCG as the metric, we further introduce two variants using PBM and UBM as the click models:
\begin{itemize}[topsep = 3pt,leftmargin =10pt]
\item \textbf{\textsf{STARank}} is the proposed framework, which uses NDCG over the whole ranking list to compute $\pi^*$ following Eq.~(\ref{eqn:oracle}).
\item \textbf{\textsf{STARank}$^+_\mathtt{PBM}$} is a variant of \textsf{STARank} using $\mathtt{R}_\mathtt{PBM}(\cdot)$ in Eq.~(\ref{eqn:cm}) to produce $\pi^*$ (i.e., PBM \citep{richardson2007predicting} as the click model).
\item \textbf{\textsf{STARank}$^+_\mathtt{UBM}$} is another variant of \textsf{STARank} using $\mathtt{R}_\mathtt{UBM}(\cdot)$ in Eq.~(\ref{eqn:cm}) to generate $\pi^*$ (i.e., UBM \citep{dupret2008user} as the click model).
\end{itemize}

We also develop the following variants to investigate the impact of each component of the proposed framework:
\begin{itemize}[topsep = 3pt,leftmargin =10pt]
\item \textbf{\textsf{STARank}$^-_\mathtt{PI}$} is a variant of \textsf{STARank} using multi-layer perceptron (MLP) layer as PI module to encode the input items instead of the attention network, where the user representation is concatenated with item embeddings to be fed into the MLP layer. 
\item \textbf{\textsf{STARank}$^-_\mathtt{PS}$}
is a variant of \textsf{STARank} using the MLP layer as PS module to encode the sequential records instead of the LSTM.
\end{itemize}

\subsection{Performance Comparison}
\label{sec:compare}
Tables~\ref{tab:res},~\ref{tab:rank}, \ref{tab:leto} and \ref{tab:seq} summarize the comparison results on the top-N recommendation datasets and learning-to-rank datasets, in terms of conventional ranking metrics and simulation-based ranking metrics.
Our major findings are listed as follows.

\minisection{\textsf{STARank} consistently outperforms all the baselines scoring and sorting.}
From these tables, we can clearly see that \textsf{STARank} outperforms all these baselines including classical tower-based methods (e.g., DeepFM), sequential methods (e.g., DIEN), and recently proposed set-to-sequence methods (e.g., SetRank) in terms of all the conventional and new proposed simulation-based ranking metrics. 
Since all these baselines are supervised by the relevance score of each item instead of directly optimizing the model by the permutation of the items in an end-to-end fashion, these results would indicate the superiority of the overall design of the proposed framework.

\minisection{\textsf{STARank} works well on the proposed simulation-based metrics.}
From comparisons between Tables~\ref{tab:res} and \ref{tab:seq}, we can see that there are larger performance gains from \textsf{STARank} in terms of the simulation-based metrics than the conventional configurations.
A possible explanation is that compared to the conventional ranking metrics, the new proposed simulation-based configurations consider the effects from the contextual dependence; and thus supervising the permutations of the ranking list would perform better in this case.

\minisection{SetRank and Seq2Slate consistently outperforms the other baselines.}
From these results, we also can observe that among 9 baseline methods, SetRank and Seq2Slate can consistently achieve the best performance.
One possible reason is that due to our experimental configuration typically assigns a set of items served as the historical records, then these methods can effectively leverage the contexts to better capture the user needs.


\begin{figure}[t]
	\centering
	\includegraphics[width=0.7\linewidth]{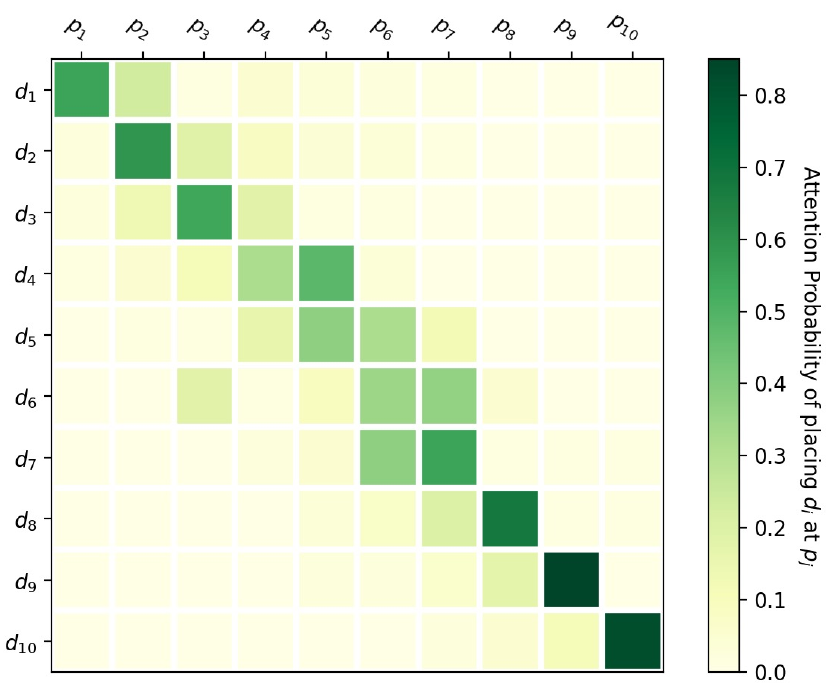}
	\vspace{-3mm}
	\caption{
	    Visualization of attention probabilities of placing item $d_i$ at position $p_j$ 
	    on Alipay dataset. 
	}
	\label{fig:attention}
	\vspace{-4mm}
\end{figure}

\subsection{Ablation Studies}
\minisection{Impact of Permutation-Invariant Module.}
To investigate the performance gain from our permutation-invariant module, we evaluate the performance of \textsf{STARank}$^-_\mathtt{PI}$ and report the corresponding results in Tables~\ref{tab:res}, \ref{tab:rank}, \ref{tab:leto}, and \ref{tab:seq}. 
From comparisons between \textsf{STARank} and \textsf{STARank}$^-_\mathtt{PI}$, the results demonstrate the superiority of our proposed attention networks introduced in Section~\ref{sec:reader} which incorporates user representation to produce the representation vector for each item.

\minisection{Impact of Permutation-Sensitive Module.}
We use \textsf{STARank}$^-_\mathtt{PS}$ to show the performance changes by replacing our RNN module with the MLP layer which is a permutation-invariant module and is not suitable to encode the permutation-sensitive user historical data. 
Results summarized in Tables~\ref{tab:res}, \ref{tab:rank}, \ref{tab:leto}, and \ref{tab:seq} verify the performance gain by using a PS module to model the user's history. 

\begin{figure}[t]
	\centering
	\includegraphics[width=0.78\linewidth]{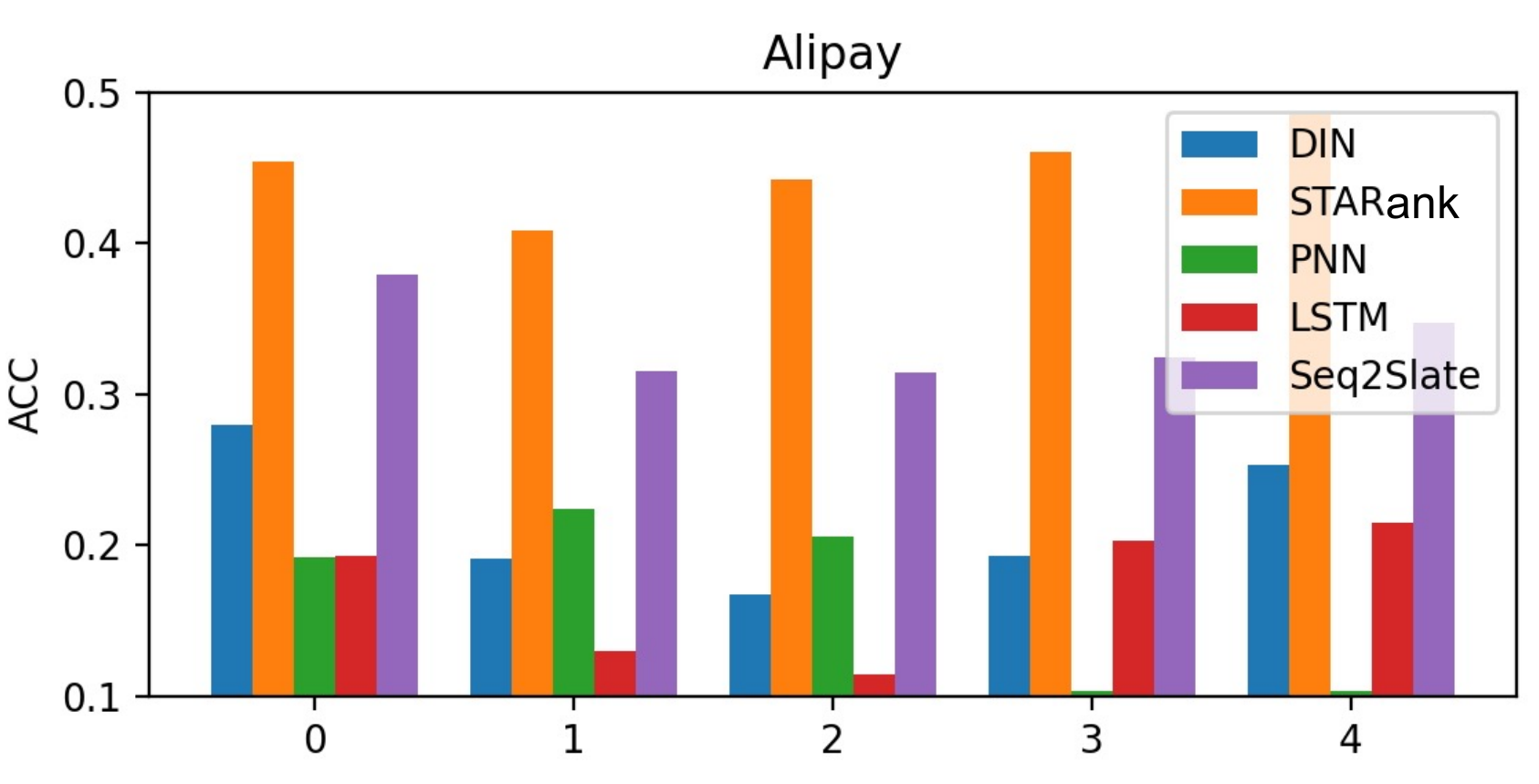}
	\vspace{-4mm}
	\caption{
	    Performance comparisons of \textsf{STARank} against baselines in terms of accuracy at top 5 positions on Alipay dataset. 
	}
	\label{fig:acc}
	\vspace{-1mm}
\end{figure}

\begin{figure}[t]
	\centering
	\includegraphics[width=0.78\linewidth]{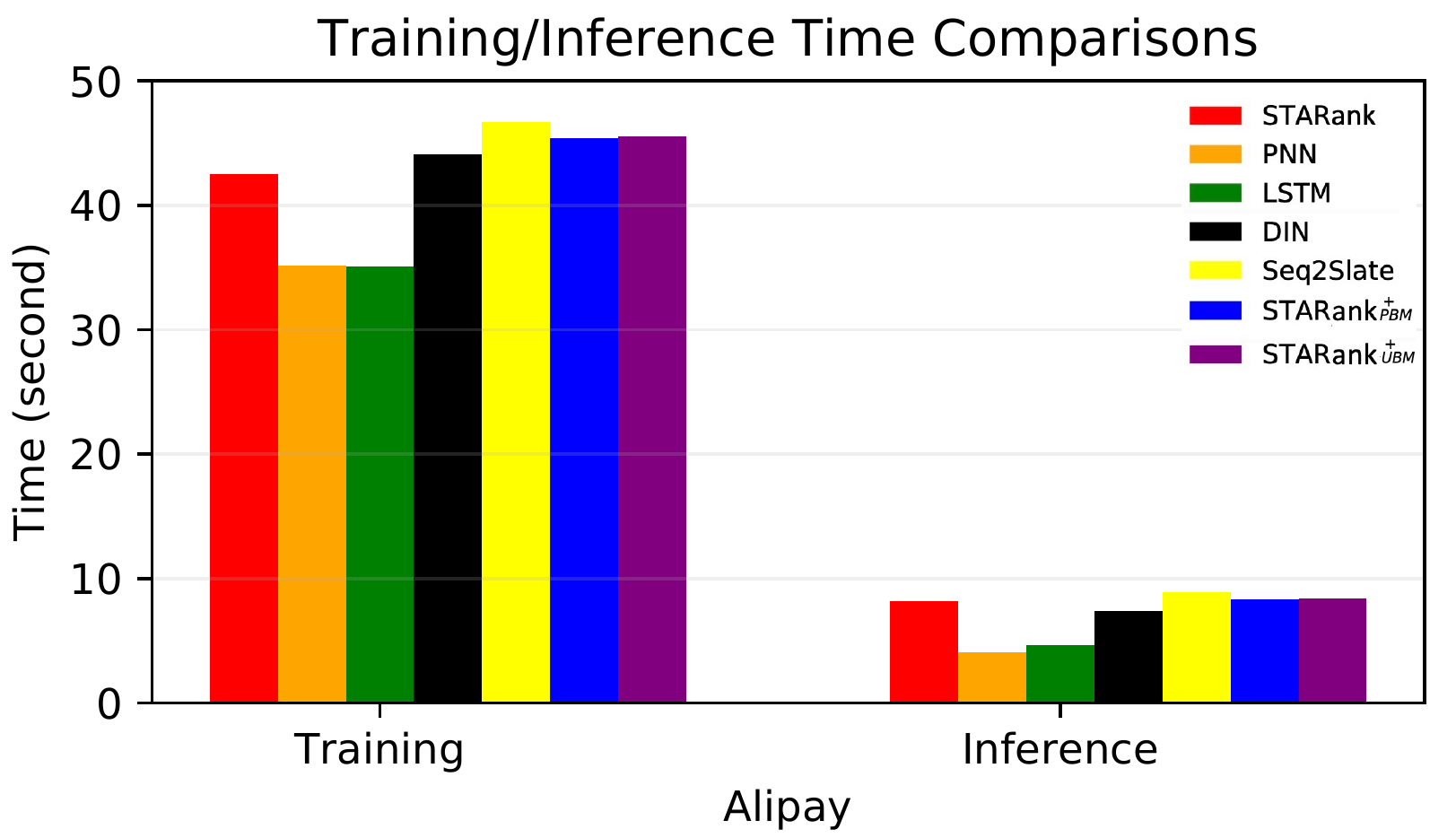}
	\vspace{-4mm}
	\caption{Comparisons of training and inference time of \textsf{STARank} and its variants against baselines on Alipay dataset. }
	\label{fig:time}
	\vspace{-2mm}
\end{figure}

\minisection{Impact of Supervision Generation.}
As described in the experimental configuration, the proposed simulation-based ranking metrics can not only be used for evaluation but also can be adopted in the ground-truth permutation generation (i.e., Eq.~(\ref{eqn:seqmetric})).
As Table~\ref{tab:seq} shows the experimental results of each model under these simulation-based ranking metrics, we further form the variants \textsf{STARank}$^+_\mathtt{PBM}$ and \textsf{STARank}$^+_\mathtt{UBM}$ whose supervision generations are simulation-based but evaluations are under conventional ranking metrics.
From the corresponding results shown in Table~\ref{tab:res}, we can see that PBM-based supervisions consistently improve the ranking performance, while UBM-based supervisions succeed in Tmall and Taobao datasets and fail in Alipay dataset.
One possible explanation is that these supervision signals encourage 
the model to encode the sequence dependence based on the assumptions of certain click models.
Therefore, if these assumptions hold in the dataset, the ranking performance would be improved; otherwise, the ranking performance would be hurt.


\subsection{In-Depth Analysis}
\minisection{Visualization Analysis.}
To better understand the behavior of the model, we visualize the probabilities of the attention scores from Eq.~(\ref{eqn:generate}), which shows the average probability of placing item $d_i$ (corresponding to $d$ in Eq.~(\ref{eqn:generate})) at position $p_j$ (corresponding to $\pi_i$ in Eq.~(\ref{eqn:generate})) at each step.
As Figure~\ref{fig:attention} depicts, the model produces the permutations that are close to the input order, but with the some items are placed to the neighboring positions.
A possible reason is that the input orders provided by Taobao platform are expected to achieve a good performance; and our method in this case can be regarded as a re-ranking model to refine the input orders.

Also, these attention values provide another perspective to explain why our framework performs better than existing ranking methods following PRP.
As revealed in \citep{dai2020u}, considering there are multiple bias in ranking, the maximization of the utility can be seen as solving the maximum-weight matching on the item-position bipartite graph, where the edge weight between an item and a position denotes the utility of placing the item at this position.
In this regard, our attention scores could be the estimation of the utility.
In contrast, PRP learning to assign individual score to each item, does not learn to estimate these the edge weights.

\minisection{Accuracy at Each Position.}
In order to precisely investigate the model performance at each position, we compute the accuracy (ACC) of the ranking list provided by PNN, LSTM, DIN, Seq2Slate, and \textsf{STARank}, and report the corresponding results on Alipay dataset in Figure~\ref{fig:acc}.
Given the predicted permutation $\pi$ and the oracle permutation $\pi^*$, for each data instance, ACC value equals 1 at position $i$ when $\pi_i=\pi^*_i$ holds; 0 otherwise.
From Figure~\ref{fig:acc}, we can see that our method can consistently outperform all these baseline methods at all the positions.
We also note that ACC values at the highest and lowest positions are higher than other positions.
One reasonable explanation is that the items that a user extremely favors or dislikes are much easier to distinguish than other items.


\subsection{Robustness Analysis}
We evaluate the robustness of \textsf{STARank} under different amounts of training data in Alipay dataset.
For comparison, we also employ LSTM as a baseline method.
Results are reported in Figure~\ref{fig:data}, which demonstrate that after reaching a threshold (60\% in this case), \textsf{STARank} is almost as robust as LSTM regarding the amounts of training data.


\section{Conclusion and Future Work}
In this paper, we present a novel set-to-arrangement framework named \textsf{STARank} to directly generate the permutations for input sets of items without the needs of scoring and sorting items.
Our architecture and supervision designs together enable \textsf{STARank} to be fully differentiable and also allow \textsf{STARank} to operate when only ground-truth permutations instead of ground-truth relevance scores for items are accessible. 
For future work, it would be interesting to investigate possible applications of the proposed framework in real-world e-commerce platforms.

\minisection{Acknowledgments.} 
    The work is partially supported by National Natural Science Foundation of China (No. 62177033) and Shanghai Artificial Intelligence Innovation and Development Fund (No. 2020-RGZN-02026).
    Jiarui Jin would like to thank Wu Wen Jun Honorary Doctoral Scholarship from AI Institute, Shanghai Jiao Tong University.

\clearpage
\bibliographystyle{ACM-Reference-Format}
\balance
\bibliography{main}


\begin{thebibliography}{39}


\ifx \showCODEN    \undefined \def \showCODEN     #1{\unskip}     \fi
\ifx \showDOI      \undefined \def \showDOI       #1{#1}\fi
\ifx \showISBNx    \undefined \def \showISBNx     #1{\unskip}     \fi
\ifx \showISBNxiii \undefined \def \showISBNxiii  #1{\unskip}     \fi
\ifx \showISSN     \undefined \def \showISSN      #1{\unskip}     \fi
\ifx \showLCCN     \undefined \def \showLCCN      #1{\unskip}     \fi
\ifx \shownote     \undefined \def \shownote      #1{#1}          \fi
\ifx \showarticletitle \undefined \def \showarticletitle #1{#1}   \fi
\ifx \showURL      \undefined \def \showURL       {\relax}        \fi
\providecommand\bibfield[2]{#2}
\providecommand\bibinfo[2]{#2}
\providecommand\natexlab[1]{#1}
\providecommand\showeprint[2][]{arXiv:#2}

\bibitem[\protect\citeauthoryear{Ai, Bi, Guo, and Croft}{Ai
  et~al\mbox{.}}{2018}]%
        {ai2018learning}
\bibfield{author}{\bibinfo{person}{Qingyao Ai}, \bibinfo{person}{Keping Bi},
  \bibinfo{person}{Jiafeng Guo}, {and} \bibinfo{person}{W~Bruce Croft}.}
  \bibinfo{year}{2018}\natexlab{}.
\newblock \showarticletitle{Learning a deep listwise context model for ranking
  refinement}. In \bibinfo{booktitle}{\emph{SIGIR}}.
\newblock


\bibitem[\protect\citeauthoryear{Bello, Kulkarni, Jain, Boutilier, Chi, Eban,
  Luo, Mackey, and Meshi}{Bello et~al\mbox{.}}{2018}]%
        {bello2018seq2slate}
\bibfield{author}{\bibinfo{person}{Irwan Bello}, \bibinfo{person}{Sayali
  Kulkarni}, \bibinfo{person}{Sagar Jain}, \bibinfo{person}{Craig Boutilier},
  \bibinfo{person}{Ed Chi}, \bibinfo{person}{Elad Eban},
  \bibinfo{person}{Xiyang Luo}, \bibinfo{person}{Alan Mackey}, {and}
  \bibinfo{person}{Ofer Meshi}.} \bibinfo{year}{2018}\natexlab{}.
\newblock \showarticletitle{Seq2slate: Re-ranking and slate optimization with
  rnns}.
\newblock \bibinfo{journal}{\emph{arXiv preprint arXiv:1810.02019}}
  (\bibinfo{year}{2018}).
\newblock


\bibitem[\protect\citeauthoryear{Burges, Shaked, Renshaw, Lazier, Deeds,
  Hamilton, and Hullender}{Burges et~al\mbox{.}}{2005}]%
        {burges2005learning}
\bibfield{author}{\bibinfo{person}{Chris Burges}, \bibinfo{person}{Tal Shaked},
  \bibinfo{person}{Erin Renshaw}, \bibinfo{person}{Ari Lazier},
  \bibinfo{person}{Matt Deeds}, \bibinfo{person}{Nicole Hamilton}, {and}
  \bibinfo{person}{Greg Hullender}.} \bibinfo{year}{2005}\natexlab{}.
\newblock \showarticletitle{Learning to rank using gradient descent}. In
  \bibinfo{booktitle}{\emph{ICML}}.
\newblock


\bibitem[\protect\citeauthoryear{Burges}{Burges}{2010}]%
        {burges2010ranknet}
\bibfield{author}{\bibinfo{person}{Christopher~JC Burges}.}
  \bibinfo{year}{2010}\natexlab{}.
\newblock \showarticletitle{From ranknet to lambdarank to lambdamart: An
  overview}.
\newblock \bibinfo{journal}{\emph{Learning}} (\bibinfo{year}{2010}).
\newblock


\bibitem[\protect\citeauthoryear{Cao, Qin, Liu, Tsai, and Li}{Cao
  et~al\mbox{.}}{2007}]%
        {cao2007learning}
\bibfield{author}{\bibinfo{person}{Zhe Cao}, \bibinfo{person}{Tao Qin},
  \bibinfo{person}{Tie-Yan Liu}, \bibinfo{person}{Ming-Feng Tsai}, {and}
  \bibinfo{person}{Hang Li}.} \bibinfo{year}{2007}\natexlab{}.
\newblock \showarticletitle{Learning to rank: from pairwise approach to
  listwise approach}. In \bibinfo{booktitle}{\emph{ICML}}.
\newblock


\bibitem[\protect\citeauthoryear{Chapelle and Chang}{Chapelle and
  Chang}{2011}]%
        {chapelle2011yahoo}
\bibfield{author}{\bibinfo{person}{Olivier Chapelle} {and} \bibinfo{person}{Yi
  Chang}.} \bibinfo{year}{2011}\natexlab{}.
\newblock \showarticletitle{Yahoo! learning to rank challenge overview}. In
  \bibinfo{booktitle}{\emph{Proceedings of the learning to rank challenge}}.
  \bibinfo{pages}{1--24}.
\newblock


\bibitem[\protect\citeauthoryear{Chuklin, Markov, and Rijke}{Chuklin
  et~al\mbox{.}}{2015}]%
        {chuklin2015click}
\bibfield{author}{\bibinfo{person}{Aleksandr Chuklin}, \bibinfo{person}{Ilya
  Markov}, {and} \bibinfo{person}{Maarten~de Rijke}.}
  \bibinfo{year}{2015}\natexlab{}.
\newblock \showarticletitle{Click models for web search}.
\newblock \bibinfo{journal}{\emph{Synthesis lectures on information concepts,
  retrieval, and services}} \bibinfo{volume}{7}, \bibinfo{number}{3}
  (\bibinfo{year}{2015}), \bibinfo{pages}{1--115}.
\newblock


\bibitem[\protect\citeauthoryear{Dai, Hou, Liu, Xi, Tang, Zhang, He, Wang, and
  Yu}{Dai et~al\mbox{.}}{2020}]%
        {dai2020u}
\bibfield{author}{\bibinfo{person}{Xinyi Dai}, \bibinfo{person}{Jiawei Hou},
  \bibinfo{person}{Qing Liu}, \bibinfo{person}{Yunjia Xi},
  \bibinfo{person}{Ruiming Tang}, \bibinfo{person}{Weinan Zhang},
  \bibinfo{person}{Xiuqiang He}, \bibinfo{person}{Jun Wang}, {and}
  \bibinfo{person}{Yong Yu}.} \bibinfo{year}{2020}\natexlab{}.
\newblock \showarticletitle{U-rank: Utility-oriented learning to rank with
  implicit feedback}. In \bibinfo{booktitle}{\emph{CIKM}}.
\newblock


\bibitem[\protect\citeauthoryear{Duan, Jiang, Qin, Zhou, and Shum}{Duan
  et~al\mbox{.}}{2010}]%
        {duan2010empirical}
\bibfield{author}{\bibinfo{person}{Yajuan Duan}, \bibinfo{person}{Long Jiang},
  \bibinfo{person}{Tao Qin}, \bibinfo{person}{Ming Zhou}, {and}
  \bibinfo{person}{Heung~Yeung Shum}.} \bibinfo{year}{2010}\natexlab{}.
\newblock \showarticletitle{An empirical study on learning to rank of tweets}.
  In \bibinfo{booktitle}{\emph{Proceedings of the 23rd International Conference
  on Computational Linguistics (Coling 2010)}}.
\newblock


\bibitem[\protect\citeauthoryear{Dupret and Piwowarski}{Dupret and
  Piwowarski}{2008}]%
        {dupret2008user}
\bibfield{author}{\bibinfo{person}{Georges~E Dupret} {and}
  \bibinfo{person}{Benjamin Piwowarski}.} \bibinfo{year}{2008}\natexlab{}.
\newblock \showarticletitle{A user browsing model to predict search engine
  click data from past observations.}. In \bibinfo{booktitle}{\emph{SIGIR}}.
\newblock


\bibitem[\protect\citeauthoryear{Friedman}{Friedman}{2001}]%
        {friedman2001greedy}
\bibfield{author}{\bibinfo{person}{Jerome~H Friedman}.}
  \bibinfo{year}{2001}\natexlab{}.
\newblock \showarticletitle{Greedy function approximation: a gradient boosting
  machine}.
\newblock \bibinfo{journal}{\emph{Annals of statistics}}
  (\bibinfo{year}{2001}), \bibinfo{pages}{1189--1232}.
\newblock


\bibitem[\protect\citeauthoryear{Guiver and Snelson}{Guiver and
  Snelson}{2009}]%
        {guiver2009bayesian}
\bibfield{author}{\bibinfo{person}{John Guiver} {and} \bibinfo{person}{Edward
  Snelson}.} \bibinfo{year}{2009}\natexlab{}.
\newblock \showarticletitle{Bayesian inference for Plackett-Luce ranking
  models}. In \bibinfo{booktitle}{\emph{proceedings of the 26th annual
  international conference on machine learning}}. \bibinfo{pages}{377--384}.
\newblock


\bibitem[\protect\citeauthoryear{Guo, Tang, Ye, Li, and He}{Guo
  et~al\mbox{.}}{2017}]%
        {guo2017deepfm}
\bibfield{author}{\bibinfo{person}{Huifeng Guo}, \bibinfo{person}{Ruiming
  Tang}, \bibinfo{person}{Yunming Ye}, \bibinfo{person}{Zhenguo Li}, {and}
  \bibinfo{person}{Xiuqiang He}.} \bibinfo{year}{2017}\natexlab{}.
\newblock \showarticletitle{DeepFM: a factorization-machine based neural
  network for CTR prediction}. In \bibinfo{booktitle}{\emph{IJCAI}}.
\newblock


\bibitem[\protect\citeauthoryear{Hidasi, Karatzoglou, Baltrunas, and
  Tikk}{Hidasi et~al\mbox{.}}{2015}]%
        {hidasi2015session}
\bibfield{author}{\bibinfo{person}{Bal{\'a}zs Hidasi},
  \bibinfo{person}{Alexandros Karatzoglou}, \bibinfo{person}{Linas Baltrunas},
  {and} \bibinfo{person}{Domonkos Tikk}.} \bibinfo{year}{2015}\natexlab{}.
\newblock \showarticletitle{Session-based recommendations with recurrent neural
  networks}.
\newblock \bibinfo{journal}{\emph{ICLR}} (\bibinfo{year}{2015}).
\newblock


\bibitem[\protect\citeauthoryear{Hochreiter and Schmidhuber}{Hochreiter and
  Schmidhuber}{1997}]%
        {hochreiter1997long}
\bibfield{author}{\bibinfo{person}{Sepp Hochreiter} {and}
  \bibinfo{person}{J{\"u}rgen Schmidhuber}.} \bibinfo{year}{1997}\natexlab{}.
\newblock \showarticletitle{Long short-term memory}.
\newblock \bibinfo{journal}{\emph{Neural computation}} \bibinfo{volume}{9},
  \bibinfo{number}{8} (\bibinfo{year}{1997}), \bibinfo{pages}{1735--1780}.
\newblock


\bibitem[\protect\citeauthoryear{Hunter}{Hunter}{2004}]%
        {hunter2004mm}
\bibfield{author}{\bibinfo{person}{David~R Hunter}.}
  \bibinfo{year}{2004}\natexlab{}.
\newblock \showarticletitle{MM algorithms for generalized Bradley-Terry
  models}.
\newblock \bibinfo{journal}{\emph{The annals of statistics}}
  (\bibinfo{year}{2004}).
\newblock


\bibitem[\protect\citeauthoryear{Joachims}{Joachims}{2002}]%
        {joachims2002optimizing}
\bibfield{author}{\bibinfo{person}{Thorsten Joachims}.}
  \bibinfo{year}{2002}\natexlab{}.
\newblock \showarticletitle{Optimizing search engines using clickthrough data}.
  In \bibinfo{booktitle}{\emph{KDD}}.
\newblock


\bibitem[\protect\citeauthoryear{Joachims}{Joachims}{2006}]%
        {joachims2006training}
\bibfield{author}{\bibinfo{person}{Thorsten Joachims}.}
  \bibinfo{year}{2006}\natexlab{}.
\newblock \showarticletitle{Training linear SVMs in linear time}. In
  \bibinfo{booktitle}{\emph{KDD}}.
\newblock


\bibitem[\protect\citeauthoryear{Joachims, Granka, Pan, Hembrooke, and
  Gay}{Joachims et~al\mbox{.}}{2017}]%
        {joachims2017accurately}
\bibfield{author}{\bibinfo{person}{Thorsten Joachims}, \bibinfo{person}{Laura
  Granka}, \bibinfo{person}{Bing Pan}, \bibinfo{person}{Helene Hembrooke},
  {and} \bibinfo{person}{Geri Gay}.} \bibinfo{year}{2017}\natexlab{}.
\newblock \showarticletitle{Accurately interpreting clickthrough data as
  implicit feedback}. In \bibinfo{booktitle}{\emph{SIGIR}}.
\newblock


\bibitem[\protect\citeauthoryear{Joachims, Granka, Pan, Hembrooke, and
  Gay}{Joachims et~al\mbox{.}}{2005}]%
        {joachims2005accurately}
\bibfield{author}{\bibinfo{person}{Thorsten Joachims}, \bibinfo{person}{Laura~A
  Granka}, \bibinfo{person}{Bing Pan}, \bibinfo{person}{Helene Hembrooke},
  {and} \bibinfo{person}{Geri Gay}.} \bibinfo{year}{2005}\natexlab{}.
\newblock \showarticletitle{Accurately interpreting clickthrough data as
  implicit feedback}. In \bibinfo{booktitle}{\emph{SIGIR}}.
\newblock


\bibitem[\protect\citeauthoryear{Lee, Lee, Kim, Kosiorek, Choi, and Teh}{Lee
  et~al\mbox{.}}{2019}]%
        {lee2019set}
\bibfield{author}{\bibinfo{person}{Juho Lee}, \bibinfo{person}{Yoonho Lee},
  \bibinfo{person}{Jungtaek Kim}, \bibinfo{person}{Adam Kosiorek},
  \bibinfo{person}{Seungjin Choi}, {and} \bibinfo{person}{Yee~Whye Teh}.}
  \bibinfo{year}{2019}\natexlab{}.
\newblock \showarticletitle{Set transformer: A framework for attention-based
  permutation-invariant neural networks}. In \bibinfo{booktitle}{\emph{ICML}}.
\newblock


\bibitem[\protect\citeauthoryear{Liu}{Liu}{2011}]%
        {liu2011learning}
\bibfield{author}{\bibinfo{person}{Tie-Yan Liu}.}
  \bibinfo{year}{2011}\natexlab{}.
\newblock \showarticletitle{Learning to rank for information retrieval}.
\newblock  (\bibinfo{year}{2011}).
\newblock


\bibitem[\protect\citeauthoryear{Pang, Xu, Ai, Lan, Cheng, and Wen}{Pang
  et~al\mbox{.}}{2020}]%
        {pang2020setrank}
\bibfield{author}{\bibinfo{person}{Liang Pang}, \bibinfo{person}{Jun Xu},
  \bibinfo{person}{Qingyao Ai}, \bibinfo{person}{Yanyan Lan},
  \bibinfo{person}{Xueqi Cheng}, {and} \bibinfo{person}{Jirong Wen}.}
  \bibinfo{year}{2020}\natexlab{}.
\newblock \showarticletitle{Setrank: Learning a permutation-invariant ranking
  model for information retrieval}. In \bibinfo{booktitle}{\emph{SIGIR}}.
\newblock


\bibitem[\protect\citeauthoryear{Plackett}{Plackett}{1975}]%
        {plackett1975analysis}
\bibfield{author}{\bibinfo{person}{Robin~L Plackett}.}
  \bibinfo{year}{1975}\natexlab{}.
\newblock \showarticletitle{The analysis of permutations}.
\newblock \bibinfo{journal}{\emph{Journal of the Royal Statistical Society:
  Series C (Applied Statistics)}} (\bibinfo{year}{1975}).
\newblock


\bibitem[\protect\citeauthoryear{Qi, Su, Mo, and Guibas}{Qi
  et~al\mbox{.}}{2017}]%
        {qi2017pointnet}
\bibfield{author}{\bibinfo{person}{Charles~R Qi}, \bibinfo{person}{Hao Su},
  \bibinfo{person}{Kaichun Mo}, {and} \bibinfo{person}{Leonidas~J Guibas}.}
  \bibinfo{year}{2017}\natexlab{}.
\newblock \showarticletitle{Pointnet: Deep learning on point sets for 3d
  classification and segmentation}. In \bibinfo{booktitle}{\emph{CVPR}}.
\newblock


\bibitem[\protect\citeauthoryear{Qin, Liu, Xu, and Li}{Qin
  et~al\mbox{.}}{2010}]%
        {qin2010letor}
\bibfield{author}{\bibinfo{person}{Tao Qin}, \bibinfo{person}{Tie-Yan Liu},
  \bibinfo{person}{Jun Xu}, {and} \bibinfo{person}{Hang Li}.}
  \bibinfo{year}{2010}\natexlab{}.
\newblock \showarticletitle{LETOR: A benchmark collection for research on
  learning to rank for information retrieval}.
\newblock \bibinfo{journal}{\emph{Information Retrieval}}
  (\bibinfo{year}{2010}).
\newblock


\bibitem[\protect\citeauthoryear{Qu, Cai, Ren, Zhang, Yu, Wen, and Wang}{Qu
  et~al\mbox{.}}{2016}]%
        {qu2016product}
\bibfield{author}{\bibinfo{person}{Yanru Qu}, \bibinfo{person}{Han Cai},
  \bibinfo{person}{Kan Ren}, \bibinfo{person}{Weinan Zhang},
  \bibinfo{person}{Yong Yu}, \bibinfo{person}{Ying Wen}, {and}
  \bibinfo{person}{Jun Wang}.} \bibinfo{year}{2016}\natexlab{}.
\newblock \showarticletitle{Product-based neural networks for user response
  prediction}. In \bibinfo{booktitle}{\emph{ICDM}}.
\newblock


\bibitem[\protect\citeauthoryear{Rendle}{Rendle}{2010}]%
        {rendle2010factorization}
\bibfield{author}{\bibinfo{person}{Steffen Rendle}.}
  \bibinfo{year}{2010}\natexlab{}.
\newblock \showarticletitle{Factorization machines}. In
  \bibinfo{booktitle}{\emph{ICDM}}.
\newblock


\bibitem[\protect\citeauthoryear{Richardson, Dominowska, and Ragno}{Richardson
  et~al\mbox{.}}{2007}]%
        {richardson2007predicting}
\bibfield{author}{\bibinfo{person}{Matthew Richardson}, \bibinfo{person}{Ewa
  Dominowska}, {and} \bibinfo{person}{Robert Ragno}.}
  \bibinfo{year}{2007}\natexlab{}.
\newblock \showarticletitle{Predicting clicks: estimating the click-through
  rate for new ads}. In \bibinfo{booktitle}{\emph{WWW}}.
\newblock


\bibitem[\protect\citeauthoryear{Robertson}{Robertson}{1977}]%
        {robertson1977probability}
\bibfield{author}{\bibinfo{person}{Stephen~E Robertson}.}
  \bibinfo{year}{1977}\natexlab{}.
\newblock \showarticletitle{The probability ranking principle in IR}.
\newblock \bibinfo{journal}{\emph{Journal of documentation}}
  (\bibinfo{year}{1977}).
\newblock


\bibitem[\protect\citeauthoryear{Sun, Vashishth, Sanyal, Talukdar, and
  Yang}{Sun et~al\mbox{.}}{2019}]%
        {sun2019re}
\bibfield{author}{\bibinfo{person}{Zhiqing Sun}, \bibinfo{person}{Shikhar
  Vashishth}, \bibinfo{person}{Soumya Sanyal}, \bibinfo{person}{Partha
  Talukdar}, {and} \bibinfo{person}{Yiming Yang}.}
  \bibinfo{year}{2019}\natexlab{}.
\newblock \showarticletitle{A re-evaluation of knowledge graph completion
  methods}.
\newblock \bibinfo{journal}{\emph{ACL}} (\bibinfo{year}{2019}).
\newblock


\bibitem[\protect\citeauthoryear{Taylor, Guiver, Robertson, and Minka}{Taylor
  et~al\mbox{.}}{2008}]%
        {taylor2008softrank}
\bibfield{author}{\bibinfo{person}{Michael Taylor}, \bibinfo{person}{John
  Guiver}, \bibinfo{person}{Stephen Robertson}, {and} \bibinfo{person}{Tom
  Minka}.} \bibinfo{year}{2008}\natexlab{}.
\newblock \showarticletitle{Softrank: optimizing non-smooth rank metrics}. In
  \bibinfo{booktitle}{\emph{WSDM}}.
\newblock


\bibitem[\protect\citeauthoryear{Vaswani, Shazeer, Parmar, Uszkoreit, Jones,
  Gomez, Kaiser, and Polosukhin}{Vaswani et~al\mbox{.}}{2017}]%
        {vaswani2017attention}
\bibfield{author}{\bibinfo{person}{Ashish Vaswani}, \bibinfo{person}{Noam
  Shazeer}, \bibinfo{person}{Niki Parmar}, \bibinfo{person}{Jakob Uszkoreit},
  \bibinfo{person}{Llion Jones}, \bibinfo{person}{Aidan~N Gomez},
  \bibinfo{person}{{\L}ukasz Kaiser}, {and} \bibinfo{person}{Illia
  Polosukhin}.} \bibinfo{year}{2017}\natexlab{}.
\newblock \showarticletitle{Attention is all you need}. In
  \bibinfo{booktitle}{\emph{NIPS}}.
\newblock


\bibitem[\protect\citeauthoryear{Vinyals, Fortunato, and Jaitly}{Vinyals
  et~al\mbox{.}}{2015}]%
        {vinyals2015pointer}
\bibfield{author}{\bibinfo{person}{Oriol Vinyals}, \bibinfo{person}{Meire
  Fortunato}, {and} \bibinfo{person}{Navdeep Jaitly}.}
  \bibinfo{year}{2015}\natexlab{}.
\newblock \showarticletitle{Pointer networks}.
\newblock \bibinfo{journal}{\emph{arXiv preprint arXiv:1506.03134}}
  (\bibinfo{year}{2015}).
\newblock


\bibitem[\protect\citeauthoryear{Xia, Liu, Wang, Zhang, and Li}{Xia
  et~al\mbox{.}}{2008}]%
        {xia2008listwise}
\bibfield{author}{\bibinfo{person}{Fen Xia}, \bibinfo{person}{Tie-Yan Liu},
  \bibinfo{person}{Jue Wang}, \bibinfo{person}{Wensheng Zhang}, {and}
  \bibinfo{person}{Hang Li}.} \bibinfo{year}{2008}\natexlab{}.
\newblock \showarticletitle{Listwise approach to learning to rank: theory and
  algorithm}. In \bibinfo{booktitle}{\emph{ICML}}.
\newblock


\bibitem[\protect\citeauthoryear{Yilmaz, Verma, Craswell, Radlinski, and
  Bailey}{Yilmaz et~al\mbox{.}}{2014}]%
        {yilmaz2014relevance}
\bibfield{author}{\bibinfo{person}{Emine Yilmaz}, \bibinfo{person}{Manisha
  Verma}, \bibinfo{person}{Nick Craswell}, \bibinfo{person}{Filip Radlinski},
  {and} \bibinfo{person}{Peter Bailey}.} \bibinfo{year}{2014}\natexlab{}.
\newblock \showarticletitle{Relevance and effort: An analysis of document
  utility}. In \bibinfo{booktitle}{\emph{CIKM}}.
\newblock


\bibitem[\protect\citeauthoryear{Zaheer, Kottur, Ravanbakhsh, Poczos,
  Salakhutdinov, and Smola}{Zaheer et~al\mbox{.}}{2017}]%
        {zaheer2017deep}
\bibfield{author}{\bibinfo{person}{Manzil Zaheer}, \bibinfo{person}{Satwik
  Kottur}, \bibinfo{person}{Siamak Ravanbakhsh}, \bibinfo{person}{Barnabas
  Poczos}, \bibinfo{person}{Ruslan Salakhutdinov}, {and}
  \bibinfo{person}{Alexander Smola}.} \bibinfo{year}{2017}\natexlab{}.
\newblock \showarticletitle{Deep sets}.
\newblock \bibinfo{journal}{\emph{arXiv preprint arXiv:1703.06114}}
  (\bibinfo{year}{2017}).
\newblock


\bibitem[\protect\citeauthoryear{Zhou, Mou, Fan, Pi, Bian, Zhou, Zhu, and
  Gai}{Zhou et~al\mbox{.}}{2019}]%
        {zhou2019deep}
\bibfield{author}{\bibinfo{person}{Guorui Zhou}, \bibinfo{person}{Na Mou},
  \bibinfo{person}{Ying Fan}, \bibinfo{person}{Qi Pi}, \bibinfo{person}{Weijie
  Bian}, \bibinfo{person}{Chang Zhou}, \bibinfo{person}{Xiaoqiang Zhu}, {and}
  \bibinfo{person}{Kun Gai}.} \bibinfo{year}{2019}\natexlab{}.
\newblock \showarticletitle{Deep interest evolution network for click-through
  rate prediction}. In \bibinfo{booktitle}{\emph{AAAI}}.
\newblock


\bibitem[\protect\citeauthoryear{Zhou, Zhu, Song, Fan, Zhu, Ma, Yan, Jin, Li,
  and Gai}{Zhou et~al\mbox{.}}{2018}]%
        {zhou2018deep}
\bibfield{author}{\bibinfo{person}{Guorui Zhou}, \bibinfo{person}{Xiaoqiang
  Zhu}, \bibinfo{person}{Chenru Song}, \bibinfo{person}{Ying Fan},
  \bibinfo{person}{Han Zhu}, \bibinfo{person}{Xiao Ma},
  \bibinfo{person}{Yanghui Yan}, \bibinfo{person}{Junqi Jin},
  \bibinfo{person}{Han Li}, {and} \bibinfo{person}{Kun Gai}.}
  \bibinfo{year}{2018}\natexlab{}.
\newblock \showarticletitle{Deep interest network for click-through rate
  prediction}. In \bibinfo{booktitle}{\emph{KDD}}.
\newblock


\end{thebibliography}

\end{document}